\documentclass[onecolumn,10pt]{IEEEtran}

\usepackage{graphicx}
\usepackage{amsmath,amssymb}
\usepackage{cases}
\usepackage{color}
\usepackage{amsfonts}
\usepackage{indentfirst}
\usepackage{slashbox}
\usepackage{cite}
\usepackage{caption}
\usepackage{algorithm}
\usepackage{algpseudocode}
\usepackage{booktabs}
\usepackage{subfigure}
\usepackage{multirow}
\usepackage{amsthm}
\makeatletter
\newtheoremstyle{mythm}{3pt}{3pt}{}{16pt}{\bfseries}{:}{.5em}{}
\theoremstyle{mythm}
\newtheorem{theorem}{Theorem}
\newtheorem{example}{Example}
\newtheorem{definition}{Definition}

\newtheorem{proposition}{Proposition}
\newtheorem{corollary}{Corollary}
\newtheorem{lemma}{Lemma}
\newtheorem{construction}{Construction}
\begin{document}
\title{On the Placement Delivery Array Design for Coded Caching Scheme in D2D Networks
\author{Jinyu Wang, Minquan Cheng, Qifa Yan, Xiaohu Tang, \IEEEmembership{Member,~IEEE}
}
\thanks{J. Wang and M. Cheng are with College of Mathematics and Statistics and Guangxi Key Lab of Multi-source Information Mining $\&$ Security respectively, Guangxi Normal University,
Guilin 541004, China (e-mail: mathwjy@163.com, chengqinshi@hotmail.com).}
\thanks{Q. Yan and X. Tang are with the Information Security and National Computing Grid Laboratory,
Southwest Jiaotong University, Chengdu, 610031, China (e-mail: qifa@my.swjtu.edu.cn, xhutang@swjtu.edu.cn).}
}
\date{}
\maketitle

\begin{abstract}
The coded caching scheme is an efficient technique as a solution to reduce the wireless network burden during the peak times in a Device-to-Device (D2D in short) communications. In a coded caching scheme, each file block should be divided into $F$ packets. It is meaningful to design a coded caching scheme with the rate and $F$ as small as possible, especially in the practice for D2D network. In this paper we first characterize coded caching scheme for D2D network by a simple array called D2D placement delivery array (DPDA in shot). Consequently some coded caching scheme for D2D network can be discussed by means of an appropriate DPDA. Secondly we derive the lower bounds on the rate and $F$ of a DPDA. According these two lower bounds, we show that the previously known determined scheme proposed by Ji et al., ( IEEE Trans. Inform. Theory, 62(2): 849-869,2016) reaches our lower bound on the rate, but does not meet the lower bound on $F$ for some parameters. Finally for these parameters, we construct three classes of DPDAs which meet our two lower bounds. Based on these DPDAs, three classes of coded caching scheme with low rate and lower $F$ are obtained for D2D network.
\end{abstract}


\section{Introduction}
Recently mainly due to asynchronous video on demand, wireless mobile data traffic is increasing rapidly. Consequently this will put an enormous strain on already-overburdened cellular networks. One of the most promising approaches for solving this problem is caching. In caching, popular contents are prefetched in the local storage across the network in off peak traffic times. During the peak traffic times, the users requests are partially served through the information stored in the local cache. So  the network load can be reduced. As a result, the caching scheme comprises of two separate phases, i.e., the content placement phase at the off peak traffic times and the content delivery phase at the peak traffic times.


In the seminal work \cite{MU}, Maddah-Ali and Niesen considered a system with a single omniscient central sever serving $K$ users through a common bottleneck link, and each user stores $M$ files from a library of $N$ files, where $0<M<N$. By jointly designing the content placement phase and the content delivery phase, a determined caching scheme, say MN scheme, was proposed to satisfy any arbitrary set of user demands with the transmission amount which reduces $1+\frac{KM}{N}$ times than that of uncoded caching scheme in the delivery phase. So far, many results have been obtained following MN schemes, for instances, \cite{AG,CYTJ,GR,STC,T,WLG,YTC} etc. Nevertheless, in \cite{MU} all communications must go through the central server, which would constitute a major bottleneck in the delivery phase. Fortunately, there is a common feature of {\em asynchronous content reuse} in video streaming applications, i.e., the same contents are requested by different users at different times \cite{NAAG}. It implies that user demands are highly redundant, then each user demand can be satisfied through local communication from a cache without traversing the central server. This communication paradigm is defined as Device-to-Device (D2D in short) communication \cite{YY}. D2D communication can highly increase the spectral efficiency of the network. In addition, it can potentially improve throughput, energy efficiency, delay, and fairness \cite{AQV}.

\subsection{Network model and problem definition}
In this paper, we focus on the caching system in D2D networks as shown in Fig. \ref{system}. Due to modeling the asynchronous content reuse and forbidding any form of uncoded multicasting gain, we also use the same assumption in \cite{MGA} that each file consists of $L$ equal blocks and each user requests $L'$ continue blocks of a file in the delivery phase.
\begin{figure}[htbp]
\centering\includegraphics[width=0.6\textwidth]{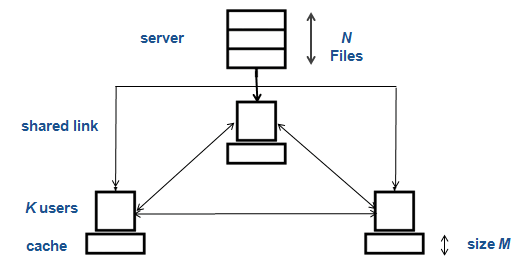}
\caption{Coded caching system with $N=3,K=3,M=1$}\label{system}
\end{figure}
And we also use the notations and the assumptions during the placement phase in \cite{MU,MGA}. That is, denote $K$ users by  $\mathcal{K}=\{0,1,\cdots,K-1\}$, and $N$ files by $\mathcal{W}=\{W_0$, $W_1$, $\ldots$, $W_{N-1}\}$. Each file is denoted by $W_{i}=(W^{0}_{i},W^{1}_{i},\ldots,W^{L-1}_{i})$, $i=$ $0$, $1$, $\ldots$, $N-1$. In order to satisfy any arbitrary demand vector, $M\geq N/K$ always holds. An $F$-division $(K,M,N,L,L')$ coded caching scheme discussed in \cite{MGA} operates in two phases:
\begin{itemize}
\item {\bf Placement Phase:} Each block is split into $F$ equal packets, i.e., $W^{l}_{i}=\{W^{l,0}_{i},W^{l,1}_{i} ,\ldots,W^{l,F-1}_{i} \}$, $i=0$, $1$, $\ldots$, $N-1$ and $l=0$, $1$, $\ldots$, $L-1$. Then the server places parts of packets in each user's cache independent of the user
demands which are assumed to be arbitrary. Denote the contents at user $k$ by $\mathcal{Z}_k$, where $k \in \mathcal{K}$.

\item{\bf Delivery Phase:} For any fixed positive integer $L'\in[1,L]$ and each integer $k\in[0,K)$, user $k$ requests $L'$ continue blocks of a random file $W_{d_k}$, i.e., $\{W{_{d_k}^{b_k}},W{_{d_k}^{b_k+1}},\cdots,W{_{d_k}^{b_k+L'-1}}\}$. Then according to the request vectors ${\bf d}=(d_0,d_1,\cdots,d_{K-1})$ and ${\bf b}=(b_0,b_1,\cdots,b_{K-1})$,
    each user broadcasts a signal to all other users using its cache contents such that each user is able to decode its requested blocks from the signals received in the delivery phase with the help of its own cache.
\end{itemize}

Denote the signal, which is broadcasted by user $k \in \mathcal{K}$, by $X_{k,{\bf d},{\bf b}}$ in the delivery phase. Suppose that the signal $X_{k,{\bf d},{\bf b}}$ is of length $S_{k,{\bf d},{\bf b}}$ packets, we define the rate $R$ (normalized by the size of $L'$ blocks) of the coded caching scheme as
\begin{equation}
\label{R}
R=\underset{{\bf d},{\bf b}}{\sup}\left\{\sum_{k=0}^{K-1}\frac{S_{k,{\bf d},{\bf b}}}{L'F}\right\}.
\end{equation}
Therefore, the first primary concern for a given $F$-division $(K,M,N,L,L')$ coded caching scheme is to minimize the rate $R$. Furthermore, the parameter $F$ may reflect the complexity of the caching scheme and it is also natural to expect it
to be as small as possible. So it is meaningful to design several classes of schemes with $R$ and $F$ as small as possible for D2D network.

\subsection{Known results and contributions}

By means of MN scheme, Ji et al. in \cite{MGA} proposed a determined $(K,M,N,L,L')$ coded caching scheme listed in Table \ref{tab-known}, which is referred to as JCM scheme, for D2D network.
\begin{table}[http]
  \centering
  \caption{Previously known determined scheme \label{tab-known}}
  \normalsize{
  \begin{tabular}{|c|c|c|c|c|c|}
\hline                                  &  $K$  & $M/N$   & $F$    & $R$ \\
\hline
 JCM scheme \cite{MGA}   &  $K$  &$\frac{M}{N}=\frac{1}{K},\ldots,\frac{K-1}{K}$
               & $\frac{KM}{N}{K\choose KM/N}$       & $\frac{N}{M}-1$\\
  \hline
   \end{tabular}}
\end{table}
They also showed that the rate of this scheme is near to the minimum rate by information theory. In this paper, we pay attention to designing the scheme with $R$ and $F$ as small as possible in D2D network. Under the same assumptions in the placement phase as in \cite{MU} and \cite{MGA}, the following main results are obtained in this paper.
\begin{itemize}
\item Firstly we propose a $L'F\times K$ array, say {\em $(K,L',F,Z,S)$ D2D placement delivery array} ($(K,L',F,Z,S)$DPDA in short), to characterize the the two phases of the $F$-division $(K,M,N,L,L')$ coded caching scheme, where $M/N=Z/F$ and $R=S/(L'F)$. As a result, designing some determined coded caching scheme for D2D network can be transformed into designing an appropriate DPDA. In fact, JCM scheme can be characterized by a special class of DPDA.
\item Secondly the lower bound $R\geq \frac{N}{M}-1$ of a DPDA is derived. This lower bound is exactly the rate of JCM scheme. Furthermore, constructing a $(K,L',F,Z,S)$ DPDA with $R=\frac{N}{M}-1$ can be reduced to constructing a $(K,1,F,Z,S')$ DPDA with the same $R$.
\item Thirdly based on this investigation, we derive the lower bound on $F$ of a DPDA with $R=\frac{N}{M}-1$ and $\frac{Z}{F}=\frac{M}{N}=\frac{1}{K}$, $\frac{2}{K}$, $1-\frac{2}{K}$, $1-\frac{1}{K}$. Consequently we show that the packet number of JCM scheme achieves the lower bound on $F$ when $\frac{M}{N}=\frac{1}{K},\frac{K-1}{K}$, but can not achieve when $\frac{M}{N}=\frac{2}{K},\frac{K-2}{K}$.
\item Finally we construct three classes of DPDA with $R=\frac{N}{M}-1$ which achieve the lower bound on $F$ for the cases $\frac{M}{N}=\frac{2}{K},\frac{K-2}{K}$.
\end{itemize}

The remainder of this paper is organized as follows. In Section \ref{PDAD2D}, the definition of DPDA is introduced and then its connection with the determined coded caching scheme is established. In Section \ref{Minimual rate}, a lower bound on the achieved rate of DPDA is derived. In Section \ref{Minimal F}, a lower bound on $F$ of DPDA ensuring the minimal rate is derived for some special cases. In Section \ref{New cosntructions}, new DPDA constructions
and their performance comparison with JCM scheme are presented. Finally, the conclusion is given in Section \ref{conclusions}.

\section{PLACEMENT DELIVERY ARRAY FOR CACHING IN WIRELESS D2D NETWORKS}
\label{PDAD2D}
In this section, we propose a new concept, i.e. {\em D2D placement delivery array} (DPDA) to characterize both the placement and delivery phases in D2D network. First, the following notations are useful

\begin{itemize}
\item $\langle a\rangle_b$ denotes the least nonnegative residue of $a$ modulo $b$ for any integer $a$ and  positive integer $b$.
\item For any nonnegative integers $s,a$ and $k \in[0,K)$, let $s^{(k)}+a={(s+a)}^{(k)}$ and $*+a=*$.
\end{itemize}

For simplicity, we also consider the identical caching policy which is assumed in the caching scheme in\cite{MU}, i.e., each
user caches the packets with the same indices from all files, that is, for any $i \in[0,F)$ and $l\in[0,L)$, the $i$-th packets of the $l$-th blocks of all files are either all stored by some user $k\in[0,K)$ or all not.  
Inspired by \cite{QMXQ}, we can use an $(LF) \times K$ array $\mathbf{P}=(p_{j,k})$, $0\leq j<LF$, $0\leq k<K$, consisting of ``$*$" and empty to characterize the placement phase for a $(K,M,N,L,L')$ coded caching scheme by the following way:
\begin{eqnarray}
\label{eq-PA}
\mathcal{Z}_k=\{ W^{l,h}_{i}\ |\  p_{j,k}=*, j\in[0,LF), l=\lfloor j/F\rfloor, h=\langle j \rangle_F,i=0,1,\ldots,N-1\}
\end{eqnarray}
where the number of ``$*$"s in each column is $LFM/N$.

\begin{example}
\label{ex-(4,2,4,2,4)PA}
Let $K=N=F=4$, $M=2,L=3$ and $L'=2$. Now let us consider the following $12\times 4$ array
\begin{eqnarray}
\label{L'FKp}
\mathbf{P}=\left(\begin{array}{cccc}
\Box   &     *    &   *    & \Box       \\
*       &   \Box   &   *    &  \Box      \\
\Box    &   *      &   \Box &   *    \\
  *     &   \Box   &   \Box &   *    \\
\Box    &     *    &   *    & \Box       \\
*       &   \Box   &   *    &  \Box      \\
\Box    &   *      &   \Box &   *    \\
  *     &   \Box   &   \Box &   *    \\
\Box    &     *    &   *    & \Box       \\
*       &   \Box   &   *    &  \Box      \\
\Box    &   *      &   \Box &   *    \\
  *     &   \Box   &   \Box &   *
\end{array}
\right),
\end{eqnarray}
where the symbol $``\Box"$ denotes an empty entry. Clearly the number of ``$*$"s in each column is
$LFM/N=3\cdot 4\cdot 2/4=6$. From \eqref{eq-PA}, a placement phase can be realized for a $(4,2,4,3,2)$ coded caching scheme by $\mathbf{P}$ in \eqref{L'FKp} as follows:
\begin{align*}
    \mathcal{Z}_0 &= \{W^{l,1}_{i},W^{l,3}_{i}\}\ |\ i\in[0,3],l\in[0,2]\}  &
    \mathcal{Z}_1 &= \{W^{l,0}_{i},W^{l,2}_{i}\}\ |\ i\in[0,3],l\in[0,2]\}\\
    \mathcal{Z}_2 &= \{W^{l,0}_{i},W^{l,1}_{i}\}\ |\ i\in[0,3],l\in[0,2]\}  &
    \mathcal{Z}_3 &= \{W^{l,2}_{i},W^{l,3}_{i}\}\ |\ i\in[0,3],l\in[0,2]\}
\end{align*}
\end{example}

In order to  make the implementation simpler and more efficient,  the identical caching is assumed for each block of all files along the lines of \cite{MGA}, i.e.,
\begin{enumerate}
\item [C$0$.] Let $p_{j,k}=*$, for any other entry $p_{j',k}$ in column $k$, if $\langle j' \rangle_F=\langle j \rangle_F$, then $p_{j',k}=*$.
\end{enumerate}

From C$0$, we have that the number of $``*"$s in each column of the first $F$ rows is the same. Clearly this number is $Z=MF/N$. So we have
\begin{enumerate}
\item [C$1$.] The symbol $ ``*"$ appears $Z$ times in each column of the first $F$ rows.
\end{enumerate}
For instance, $\mathbf{P}$ in \eqref{L'FKp} satisfies properties C$0$ and C1. 
It is very interesting that an array can also be used to characterize the delivery phase for a $(K,M,N,L,L')$ coded caching scheme when we properly place each required empty entry, whose row index is corresponding to some packet which is needed by the user corresponding to the column index, by some integer in the set $\{s^{(k_s)}| s\in[0,S), k_s\in[0,K)\}$, where $S$ is a positive integer, such that if the obtained array satisfies C$0$, C$1$ and the following properties:
\begin{enumerate}
  \item [C2.] For any $s\in [0,S)$, $s^{(k_s)}$ occurs at least once in the array;
  \item [C3.] If $p_{j,k}=s^{(k_s)},s\in [0,S)$, then $p_{j,k_s}=*$;
  \item [C4.] For any two distinct entries $p_{j_1,k_1}$ and $p_{j_2,k_2}$, if   $p_{j_1,k_1}=p_{j_2,k_2}=s^{(k_s)},s\in [0,S)$,  then
  \begin{enumerate}
     \item [a.] $j_1\ne j_2$, $k_1\ne k_2$, i.e., they lie in distinct rows and distinct columns;
     \item [b.] $p_{j_1,k_2}=p_{j_2,k_1}=*$, i.e., the corresponding $2\times 2$  subarray formed by rows $j_1,j_2$ and columns $k_1,k_2$ must be of the following form

    $$
    \begin{pmatrix}
    s^{(k_s)} & *         \\
    *         & s^{(k_s)}
    \end{pmatrix}
    \qquad   \text{or}  \qquad
    \begin{pmatrix}
    *         & s^{(k_s)}    \\
    s^{(k_s) }& *
    \end{pmatrix}
    $$

  \end{enumerate}
\end{enumerate}

\begin{example}
\label{ex-(4,2,4,2,4)DPDA}
We continue to consider Example \ref{ex-(4,2,4,2,4)PA}. Assume that the request vectors are ${\bf d}=(0,1,2,3)$
and ${\bf b}=(0,1,0,1)$ in the delivery phase. Through properly choosing each required empty entry from the set $\{s^{(k_s)}|s\in [0,8),k_s\in [0,4)\}$, 
$\mathbf{P}$ in \eqref{L'FKp} can be written as
\begin{eqnarray}
\label{LFKpd}
\mathbf{P}=\left(\begin{array}{cccc}
2^{(2)} &     *    &   *     & \Box       \\
*       &   \Box   &   *     &  \Box      \\
3^{(3)} &   *      & 1^{(1)} &   *    \\
  *     &   \Box   &  0^{(0)}&   *    \\
\hline
6^{(2)} &     *    &   *     & 1^{(1)}      \\
*       & 2^{(2)}  &   *     & 0^{(0)}    \\
7^{(3)} &   *      & 5^{(1)} &   *    \\
  *     & 3^{(3)}  & 4^{(0)} &   *    \\
\hline
\Box    &     *    &   *     & 5^{(1)}     \\
*       & 6^{(2)}  &   *     & 4^{(0)}    \\
\Box    &   *      &   \Box  &   *    \\
  *     & 7^{(3)}  &   \Box  &   *
\end{array}
\right).
\end{eqnarray}
It is not difficult to check that above $\mathbf{P}$ satisfies C2, C3, and C4 too. According to \eqref{LFKpd}, the coded multicast messages are as follows:
\begin{align*}
    X_{0,{\bf d},{\bf b}} &=W{_{2}^{0,3}}\oplus W{_{3}^{1,1}}, W{_{2}^{1,3}}\oplus W{_{3}^{2,1}} &
    X_{1,{\bf d},{\bf b}} &=W{_{2}^{0,2}}\oplus W{_{3}^{1,0}}, W{_{2}^{1,2}}\oplus W{_{3}^{2,0}}\\
    X_{2,{\bf d},{\bf b}} &=W{_{0}^{0,0}}\oplus W{_{1}^{1,1}}, W{_{0}^{1,0}}\oplus W{_{1}^{2,1}} &
    X_{3,{\bf d},{\bf b}} &=W{_{0}^{0,2}}\oplus W{_{1}^{1,3}}, W{_{0}^{1,2}}\oplus W{_{1}^{2,3}}
\end{align*}

Then for user $0$, it requests the first and second blocks of file $W_0$, i.e. $\{W{_0^0},W{_0^1}\}$, and it has cached packets $\{1,3\}$ of each block, so it needs packets $\{0,2\}$ of the first and second blocks of file $W_0$, i.e. $\{W{_0^{0,0}},W{_0^{0,2}},W{_0^{1,0}},W{_0^{1,2}}\}$. And user $2$ broadcasts $W{_{0}^{0,0}}\oplus W{_{1}^{1,1}}$ and $W{_{0}^{1,0}}\oplus W{_{1}^{2,1}}$, user $0$ can decode $W{_{0}^{0,0}}$ and $W{_{0}^{1,0}}$ because it has stored $W{_{1}^{1,1}}$ and $W{_{1}^{2,1}}$ in its cache. Moreover, user $3$ broadcasts $W{_{0}^{0,2}}\oplus W{_{1}^{1,3}}$ and $W{_{0}^{1,2}}\oplus W{_{1}^{2,3}}$, user $0$ can decode $W{_{0}^{0,2}}$ and $W{_{0}^{1,2}}$ because it has stored $W{_{1}^{1,3}}$ and $W{_{1}^{2,3}}$. The analysis is similar for other users. In a word, each user can decode its requested blocks correctly with the help of the contents within its own cache. Furthermore, users broadcast collectively $S=8$ messages, each of length one packet, hence, the rate of the caching scheme is $R=S/(L'F)=1$.
\end{example}

From Example \ref{ex-(4,2,4,2,4)DPDA}, we find out that the following $8\times 4$ array
\begin{eqnarray}
\label{L'FKpd}
\mathbf{Q}=\left(\begin{array}{cccc}
2^{(2)} &     *    &   *    & 1^{(1)}\\
*       & 2^{(2)}  &   *    & 0^{(0)}\\
3^{(3)} &   *      & 1^{(1)}&   * \\
  *     & 3^{(3)}  & 0^{(0)}&   * \\
\hline
6^{(2)} &     *    &   *    & 5^{(1)}\\
*       & 6^{(2)}  &   *    & 4^{(0)}\\
7^{(3)} &   *      & 5^{(1)}&   * \\
  *     & 7^{(3)}  & 4^{(0)}&   *
\end{array}
\right),
\end{eqnarray}
whose entries are in the set $\{*\}\cup \{s^{(k_s)}|s\in [0,8),k_s\in [0,4)\}$ also satisfies C0,C1,C2,C3 and C4. Moreover, it can also be used to realize the same placement and delivery phase realized by $\mathbf{P}$ in \eqref{LFKpd}. Obviously, $\mathbf{Q}$ is a $(L'F)\times K$ array, whose number of rows is much smaller than $LF$ of $\mathbf{P}$ when $L$ is far greater than $L'$. Moreover, It's easier to determine the entries in $\mathbf{Q}$ than in $\mathbf{P}$ because of the existence of empty entries in $\mathbf{P}$. So we call $\mathbf{Q}$ in \eqref{L'FKpd} a {\em D2D placement delivery array} for convenient.

\begin{definition}
For  positive integers $K,L',F,Z$ and $S$, a $(L'F)\times K$ array  $\mathbf{P}=(p_{i,j})$, $i\in [0,L'F), j\in[0,K)$ whose entries are in the set $\{*\}\cup \{s^{(k_s)}|s\in [0,S),k_s\in [0,K)\}$, is called a $(K,L',F,Z,S)$ {\em D2D placement delivery array}, or $(K,L',F,Z,S)$ DPDA for short, if it satisfies C0, C1, C2, C3 and C4.
\end{definition}

\begin{theorem}
\label{thDPDA}
A $(K,M,N,L,L')$ coded caching scheme for D2D network can be realized by a $(K,L',F,Z,S)$ DPDA $\mathbf{P}=(p_{i,j})_{(L'F) \times K}$ with $Z/F=M/N$. Precisely, each user can decode its requested blocks correctly for any request vectors ${\bf d}$ and ${\bf b}$ at the rate $R=S/(L'F)$.
\end{theorem}

\begin{proof}
The coded caching scheme can be implemented according to $\mathbf{P}$ as follows:
\begin{itemize}
\item {\bf Placement Phase:}
Each user's cache can be obtained by \eqref{eq-PA} and properties C$0$, C$1$ directly, i.e.,
\begin{equation}
\label{eqcache}
    \mathcal{Z}_j=\{W{_{i}^{l,h}}|p_{h,j}=*, h \in [0,F), i \in [0,N), l\in [0,L)\}
\end{equation}
Then each user caches $ZLN$ packets, i.e. $ZLN/(LF)=M$ files,  which satisfies the users’ cache constraint.

\item {\bf Delivery Phase:}
Assume that the request vectors are ${\bf d}=(d_0,d_1,\cdots,d_{K-1})$ and ${\bf b}=(b_0,b_1,\cdots, b_{K-1})$. Then at the time slot $s$, $s\in[0,S)$, user $k_s$ broadcasts
\begin{equation}
\label{eqdiliver}
    \underset{p_{i,j}=s^{(k_s)},i\in [0,L'F),j\in[0,K)}{\bigoplus}  W{_{d_j}^{b_j+\lfloor \frac{i}{F} \rfloor,\langle i \rangle_F}}.
\end{equation}
\end{itemize}
We claim that each user can decode its requested blocks correctly for any request vectors ${\bf d}$ and ${\bf b}$ at the rate $R=S/(L'F)$. In fact, for any $s \in [0,S)$, assume that $s^{(k_s)}$ occurs $g_s$ times in $\mathbf{P}$, and define
$$\{(i,j)|p_{i,j}=s^{(k_s)},i\in [0,L'F), j\in [0,K) \}=\{(i_1,j_1),(i_2,j_2),\cdots, (i_{g_s},j_{g_s}) \}.$$
Then we have $p_{i_h,k_s}=*$ for any $h\in [1,g_s]$ by C3 and $p_{i_u,j_v}=*$ for all
$1\leq u\neq v\leq g_s$ by C4. So the $g_s \times (g_s+1)$ subarray formed by rows $i_1,i_2,\cdots,i_{g_s}$ and columns $j_1,j_2,\cdots,j_{g_s},k_s$ is equivalent to the following array
\begin{equation}
\label{eqsubarray}
\bordermatrix{%
&j_1 & j_2 & \cdots & j_{g_s} & k_s \cr
i_1    & s^{(k_s)} & *         & \cdots & *        & *      \cr
i_2    & *         & s^{(k_s)} & \cdots & *        & *      \cr
\vdots & \vdots    & \vdots    & \ddots & \vdots   & \vdots \cr
i_{g_s}& *         & *         & \cdots &s^{(k_s)} &*}
\end{equation}
with respect to row/column permutation. According to \eqref{eqdiliver}, at the time slot $s$, user $k_s$ broadcasts
\begin{equation}
\label{eqsignal}
    \underset{1\leq h \leq g_s}{\bigoplus}  W{_{d_{j_h}}^{b_{j_h}+\lfloor \frac{i_h}{F} \rfloor,\langle i_h \rangle_F}}.
\end{equation}
Note from \eqref{eqsubarray} that in column $k_s$, all entries are $``*"$s, it follows \eqref{eqcache} and C$0$  that user $k_s$ has already all the packets $\{W{_{d_{j_h}}^{b_{j_h}+\lfloor \frac{i_h}{F} \rfloor,\langle i_h \rangle_F}}| h\in[1,g_s]\}$ in its cache. So it can broadcast the signal \eqref{eqsignal}. In column $j_h$, $h\in[1,g_s]$, all entries are $``*"$s except for the $h$-th one, it follows \eqref{eqcache} and C$0$ that user $j_h$ has cached all other packets $\{W{_{d_{j_u}}^{b_{j_u}+\lfloor \frac{i_u}{F} \rfloor,\langle i_u \rangle_F}}|1\leq u \neq h \leq g_s\}$ except $W{_{d_{j_h}}^{b_{j_h}+\lfloor \frac{i_h}{F} \rfloor,\langle i_h \rangle_F}}$. So it can easily decode the desired packet $W{_{d_{j_h}}^{b_{j_h}+\lfloor \frac{i_h}{F} \rfloor,\langle i_h \rangle_F}}$.

Since users broadcast collectively $S$ packets for each possible request vectors ${\bf d}$ and ${\bf b}$, the rate of the scheme is given by $R=\frac{S}{L'F}$.
 \end{proof}

From Theorem \ref{thDPDA}, a $(K,L',F,Z,S)$ DPDA with $Z/F=M/N$ can realize both placement and delivery phase of a $(K,M,N,L,L')$ coded caching scheme for D2D network in a single array. This implies that the problem of finding some coded caching scheme could be translated into designing an appropriate DPDA. In fact, JCM scheme can be realized by a special class of DPDA.

\begin{theorem}
\label{JCMtoDPDA}
Any $(K,M,N,L,L')$ JCM scheme is corresponding to a $(K,L',F,Z,S)$ DPDA, where $F= t{K\choose t}$, $Z=t{K-1\choose t-1}$, $S=L'(t+1){K\choose t+1}$, $t=KM/N$ assumed to be an integer.
\end{theorem}

The proof of Theorem \ref{JCMtoDPDA} is included in Appendix A. In the next section, we will discuss the minimal rate of DPDA.

\section{Minimal rate of DPDA}
\label{Minimual rate}
As mentioned before, the rate of a coded caching scheme realized by a $(K,L',F,Z,S)$ DPDA is $S/(L'F)$. Clearly it is very meaningful to make the value of $S/(L'F)$ as small as possible.
\begin{theorem}
\label{opdpda}
For any $(K,L',F,Z,S)$ DPDA $\mathbf{P}$, the rate of the coded caching scheme of $\mathbf{P}$ satisfies
\begin{equation*}
 R=\frac{S}{L'F} \geq \frac{F}{Z}-1,
\end{equation*}
the equality holds if and only if
\begin{enumerate}
  \item [C$2'$] for any $s \in [0 , S)$, $s^{(k_s)}$ appears $\frac{KZ}{F}$ times and
  \item [C$5$] each row has exactly $\frac{KZ}{F}$ $``*"$s.
\end{enumerate}

\end{theorem}
\begin{proof}
For any $s\in [0,S)$, assuming that $s^{(k_s)}$ appears $r_s$ times, there exists a subarray equivalent to the following $r_s \times  (r_s+1)$ array
\begin{eqnarray}
\label{eqsubarray1}
\left(\begin{array}{ccccc}
s^{(k_s)} & *         & \cdots & *        & *      \\
*         & s^{(k_s)} & \cdots & *        & *      \\
\vdots    & \vdots    & \ddots & \vdots   & \vdots \\
*         & *         & \cdots &s^{(k_s)} &*
\end{array}
\right)
\end{eqnarray}
with respect to row/column permutation by \eqref{eqsubarray}. Obviously there are $r{_s^2}$ $``*"$s in \eqref{eqsubarray1}. This implies that for each integer $s$ there are $r{_s^2}$ $``*"$s used in \eqref{eqsubarray1}. Hence, for all the integers in $\mathbf{P}$, there are totally $\sum_{s=0}^{S-1}r{_s^2}$ $``*"$s used.

On the other hand, for any $i \in [0,L'F)$, assuming that there are $t_i$ integers in row $i$, then the remaining $(K-t_i)$ entries are $``*"$s. So the number of $``*"$s in row $i$, which can be used by the $t_i$ integers, is no more than $(K-t_i)$. Then the $``*"$s in row $i$ are used no more than $t_i(K-t_i)$ times.  Hence, all the $``*"$s in $\mathbf{P}$ are totally used no more than $\sum_{i=0}^{L'F-1}t_i(K-t_i)$ times. That is
\begin{equation*}
\sum_{s=0}^{S-1}r{_s^2} \leq \sum_{i=0}^{L'F-1}t_i(K-t_i).
\end{equation*}
This implies that
\begin{equation}
\label{eqxs}
\sum_{s=0}^{S-1}r{_s^2}+\sum_{i=0}^{L'F-1}t{_i^2} \leq \sum_{i=0}^{L'F-1}t_iK=K^2(L'F-L'Z).
\end{equation}
Moreover
\begin{equation}
\label{eqbds1}
\begin{split}
\sum_{s=0}^{S-1}r{_s^2} &\geq \frac{1}{S}\left(\sum_{s=0}^{S-1}r_s\right)^2=\frac{1}{S}K^2(L'F-L'Z)^2,\\
\sum_{i=0}^{L'F-1}t{_i^2} &\geq \frac{1}{L'F}\left(\sum_{i=0}^{L'F-1}t_i\right)^2=\frac{1}{L'F}K^2(L'F-L'Z)^2,
\end{split}
\end{equation}
and the the above equalities hold if and only if $r_0=r_1=\cdots=r_{S-1}$ and $t_0=t_1=\cdots=t_{L'F-1}$. Submitting \eqref{eqbds1} into \eqref{eqxs}, we have
\begin{equation*}
\frac{L'F-L'Z}{S} \leq \frac{Z}{F},\ \ \ \ \   \hbox{i.e.,}\ \ \ \ \ \frac{1-\frac{Z}{F}}{\frac{Z}{F}} \leq \frac{S}{L'F}.
\end{equation*}
So $R=\frac{S}{L'F} \geq \frac{1-\frac{Z}{F}}{\frac{Z}{F}}=\frac{F}{Z}-1$ and
the equality holds if and only if the number of $``*"$s in each row is a constant integer, denoted by $t$, and each integer in $[0,S)$ occurs the same times, denoted by $g$. Since there are $L'ZK$ ``$*$" entries and $L'(F-Z)K$ integer entries, then $t=L'ZK/(L'F)=KZ/F$ and $g=L'(F-Z)K/S=K(1-\frac{Z}{F})/(\frac{S}{L'F})=K(1-\frac{Z}{F})/(\frac{F}{Z}-1)=KZ/F$.
\end{proof}

From Table \ref{tab-known} and Theorem \ref{opdpda}, it is easy to check that JCM scheme reaches this lower bound. In fact, based on a $(K,1,F,Z,S)$ DPDA $\mathbf{P}$ with $Z/F=M/N$, a $(K,L',F,Z,L'S)$ DPDA $\mathbf{Q}$ can be constructed as follows
\begin{eqnarray}
\label{DPDA1toL'}
\mathbf{Q} =\left(\begin{array}{c}
\mathbf{P}     \\
\mathbf{P}+S\cdot {\bf J}_{F\times K}   \\
\mathbf{P}+2S\cdot {\bf J}_{F\times K}  \\
\vdots         \\
\mathbf{P}+(L'-1)S\cdot {\bf J}_{F\times K}
\end{array}
\right)
\end{eqnarray}
where ${\bf J}_{F\times K}$ is an $F\times K$ array with each entry $1$. So we have
\begin{corollary}
\label{th1toL'}
Let $\mathbf{P}$ be a $(K,1,F,Z,S)$ DPDA, then $\mathbf{Q}$ in \eqref{DPDA1toL'} is a $(K,L',F,Z,L'S)$ DPDA.
\end{corollary}

From Theorem \ref{opdpda} we know that the minimal rate of a $(K,L',F,Z,S)$ DPDA with $Z/F=M/N$ is $R_{min}=N/M-1$. And it follows Corollary \ref{th1toL'} that a $(K,L',F,Z,L'S)$ DPDA $\mathbf{Q}$ can be easily constructed from a $(K,1,F,Z,S)$ DPDA $\mathbf{P}$ while keeping the rate unchanged.  It implies that if a $(K,1,F,Z,S)$ DPDA $\mathbf{P}$ with $Z/F=M/N$ achieves the minimal rate $R=R_{min}=N/M-1$, then the $(K,L',F,Z,L'S)$ DPDA $\mathbf{Q}$ in \eqref{DPDA1toL'} also achieves the minimal rate. In this sense, we only need to consider the case of $L'=1$ in the remainder of this paper.

\section{Minimal $F$ of DPDA with minimal rate}
\label{Minimal F}
From this section, we will focus on the case of $L'=1$. First the following theorem is useful.

\begin{theorem}
\label{userbroadcast}
Given a $(K,1,F,Z,S)$ DPDA $\mathbf{P}$, if $R=\frac{S}{F}=\frac{F}{Z}-1$, then for any $k\in[0,K)$, the number of integers, whose superscript is $(k)$, is
\begin{equation}
\label{broadcasttimes}
m_k=\frac{F}{K}\left(\frac{F}{Z}-1\right).
\end{equation}
That is, each user broadcasts the same times.
\end{theorem}

\begin{proof}
Since $R=\frac{F}{Z}-1$, from Theorem \ref{opdpda} we have that each integer appears $t=\frac{KZ}{F}$ times, and each row has exactly $t$ $``*"$s. This implies that there are $(K-t)$ integers in each row and each $``*"$ is used by $(K-t)$ integers. So each $``*"$ in $\mathbf{P}$ is used $(K-t)$ times. Clearly the number of $``*"$s in each column is $Z$. So the $``*"$s in each column are totally used $(K-t)Z$ times.

On the other hand, for any $k\in[0,K)$, assume that the number of integers, whose superscript is $(k)$, is $m_k$. Then the $``*"$s in column $k$ are used $tm_k$ times. Moreover, there are $(F-Z)$ integers in column $k$, the $``*"$s in column $k$ are used $(F-Z)(t-1)$ times again by the above discussion. So the $``*"$s in column $k$ are totally used $tm_k+(F-Z)(t-1)$ times. So we have
\begin{equation*}
(K-t)Z=tm_k+(F-Z)(t-1).
\end{equation*}
Then we have
\begin{equation*}
m_k=\frac{F}{K}(\frac{F}{Z}-1).
\end{equation*}
\end{proof}

From Theorem \ref{userbroadcast}, we know that $m_k$ is an constant independent of $k$. In the remainder of this section, we will give the lower bound of $F$ in a $(K,1,F,Z,S)$ DPDA ensuring the minimal rate $R=\frac{F}{Z}-1$ in four cases: $\frac{Z}{F}=\frac{1}{K}, \frac{2}{K}, \frac{K-2}{K}, \frac{K-1}{K}$.
\begin{lemma}
\label{le-opg1}
If $\mathbf{P}$ is a $(K,1,F,Z,S)$ DPDA with $\frac{Z}{F}=\frac{1}{K}$ and $R=\frac{F}{Z}-1$, then $F \geq K$.
\end{lemma}

\begin{proof}
Since $R=\frac{F}{Z}-1$ is the minimal rate, we have that each row of $\mathbf{P}$ has exactly $t=\frac{KZ}{F}=1$ $``*"$ from Theorem \ref{opdpda}. Since $Z=\frac{F}{K}\geq 1$, $F\geq K$ always holds.
\end{proof}

\begin{example}
From Theorem \ref{opdpda} and Lemma \ref{le-opg1}, it is easy to check that the following $(3,1,3,1,6)$ DPDA $\mathbf{P}$ achieves the minimal rate and the minimal $F$. Actually, it has the same parameters of JCM scheme \cite{MGA} listed in Table \ref{tab-known}.
\begin{eqnarray*}
\mathbf{P} =\left(\begin{array}{ccc}
 *      &  0^{(0)}&  1^{(0)}\\
 3^{(1)}&  *      &  2^{(1)}\\
 4^{(2)}& 5^{(2)} &  *
\end{array}
\right)
\end{eqnarray*}
\end{example}

\begin{lemma}
\label{le-opgK-1}
If $\mathbf{P}$ is a $(K,1,F,Z,S)$ DPDA with $\frac{Z}{F}=\frac{K-1}{K}$ and $R=\frac{F}{Z}-1$, then $F \geq K(K-1)$.
\end{lemma}
\begin{proof}
Since $R=\frac{F}{Z}-1$ is the minimal rate, from Theorem \ref{userbroadcast} we have that for any $k\in[0,K)$, the number of integers, whose superscript is $(k)$, is $m_k=\frac{F}{K}(\frac{F}{Z}-1)=\frac{F}{K(K-1)}$. Since $m_k\geq 1$, we have $F \geq K(K-1)$.
\end{proof}

\begin{example}
From Theorem \ref{opdpda} and Lemma \ref{le-opgK-1}, it is easy to check that the following $(3,1,6,4,3)$ DPDA $\mathbf{P}$ achieves the minimal rate and the minimal $F$.
\begin{eqnarray}
\label{(3,1,6,4,3)DPDA}
\mathbf{P}=\left(\begin{array}{ccc}
     *     &     *     & 0^{(0)}  \\
     *     &  0^{(0)}  &     *    \\
  1^{(1)}  &     *     &     *    \\
     *     &     *     &  1^{(1)} \\
     *     &   2^{(2)} &     *    \\
   2^{(2)} &   *       &     *
   \end{array}
\right)
\end{eqnarray}

\end{example}
\begin{lemma}
\label{le-opg2}
If $\mathbf{P}$ is a $(K,1,F,Z,S)$ DPDA with $\frac{Z}{F}=\frac{2}{K}$ and $R=\frac{F}{Z}-1$, then $F \geq \frac{K^2}{4}$.
\end{lemma}

\begin{proof}
Since $R=\frac{F}{Z}-1$, C$2'$ and C$5$ hold by Theorem \ref{opdpda}, i.e., each integer appears $t=\frac{KZ}{F}=2$ times, and each row has exactly two $``*"$s. After exchange of two rows/columns with appropriate integers' superscript adjustment, we can always assume that the first two rows of $\mathbf{P}$ is as follows, where the superscript of all integers are omitted for simplicity.
\begin{eqnarray*}
\begin{array}{c@{\hspace{-5pt}}l}
 \left(\begin{array}{ccc|cccc}
 *      &  0  &  *    &   a_0 &  a_1 & \cdots & a_{K-4}  \\
 *      &  *  &  0    &   b_0 &  b_1 & \cdots & b_{K-4} \\\hline
        &     &       &       &      &        &         \\
        &     &       &       &      &        &        \\
\multicolumn{3}{c|}%
          {\raisebox{2ex}[0pt]{\LARGE $\mathbf{A_1}$}}  &
\multicolumn{4}{c}%
          {\raisebox{2ex}[0pt]{\LARGE $\mathbf{A_2}$}}
 \end{array}\right)
\end{array}
\end{eqnarray*}
By C$4$, we have that (i) $a_0,\cdots,a_{K-4},b_0,\cdots,b_{K-4}$ are different integers in $[1,S)$; (ii) $a_i$ and $b_i$ can't appear in $\mathbf{A_2}$ for any $i \in [0,K-4]$. So $a_i$ and $b_i$ appears exactly once in $\mathbf{A_1}$ by C$2'$. Moreover, from C$4$ we have that each $a_i$ has to appear in column zero or two of $\mathbf{A_1}$, and each $b_i$ has to appear in column zero or one of $\mathbf{A_1}$. Now let us consider the entries in $\{a_0,\cdots,a_{K-4}\}$ first.
\begin{itemize}
  \item If $a_i=p_{j,0}$ for some $i\in[0,K-4]$ and $j\in[2,F)$, it is easy to check that $$p_{j,2}=*\ \ \ \ \hbox{and} \ \ p_{j,1}\in [0,S)\setminus\{0,a_0,\cdots,a_{K-4},b_0,\cdots,b_{K-4}\}$$
      hold by C$3$, C$4$ and C$5$. Assume the number of such integers $j$ is $h_1$.
  \item If $a_i=p_{j,2}$ for some $i\in[0,K-4]$ and $j\in[2,F)$, then $p_{j,0}=*$ by C$3$. And we claim that $$p_{j,1}\in[0,S)\setminus\{0,a_0,\cdots,a_{K-4}\}$$ by C$4$ and C$5$.
      There are two cases: (i) $p_{j,1}=b_{i'}$ for some $i'\in[0,K-4]$, then $i'=i$. Otherwise if $i'\neq i$, there are at least three $``*"$s in row $j$ by C$4$, which contradicts with C$5$. Assume the number of such integers $j$ is $h_2$. (ii) $p_{j,1}\notin \{b_0,\cdots,b_{K-4}\}$, assume the number of such integers $j$ is $h_3$.
  \end{itemize}
  Similar to the above discussion, we can also obtain the following results.
  \begin{itemize}
  \item If $b_i=p_{j,0}$ for some $i\in[0,K-4]$ and $j\in[2,F)$, it is easy to check that $$p_{j,1}=*\ \ \ \ \hbox{and} \ \ p_{j,2}\in [0,S)\setminus\{0,a_0,\cdots,a_{K-4},b_0,\cdots,b_{K-4}\}$$
      hold by C$3$, C$4$ and C$5$. Assume the number of such integers $j$ is $h_4$.
  \item If $b_i=p_{j,1}$ for some $i\in[0,K-4]$ and $j\in[2,F)$, then $p_{j,0}=*$ by C$3$. And we claim that $$p_{j,2}\in[0,S)\setminus\{0,b_0,\cdots,b_{K-4}\}$$ by C$4$ and C$5$. If $p_{j,1}=a_{i'}$ for some $i'\in[0,K-4]$, then $i'=i$, and the number of such integers $j$ was assumed to be $h_2$ above. If $p_{j,1}\notin \{a_0,\cdots,a_{K-4}\}$, assume the number of such integers $j$ is $h_5$.
\end{itemize}
From above discussion, $\mathbf{P}$ can be written in the following way.
\begin{equation*}
\label{2DPDA}
\begin{array}{c@{\hspace{-5pt}}l}
 \begin{array}{l}
    h_1 \left\{\rule{0mm}{2mm}\right. \\
     \\
    h_2 \left\{\rule{0mm}{2mm}\right.\\
    \\
    h_3 \left\{\rule{0mm}{2mm}\right.\\
    \\
    h_4 \left\{\rule{0mm}{2mm}\right.\\
    \\
    h_5 \left\{\rule{0mm}{2mm}\right.\\
   \end{array}
 &\left(\begin{array}{ccc|cccc}
 *      &  0  &  *    &   a_0 &  a_1 & \cdots & a_{K-4}  \\
 *      &  *  &  0    &   b_0 &  b_1 & \cdots & b_{K-4} \\\hline
 a_i    &\Box &   *   &       &      &        &        \\
        &     &       &       &      &        &         \\
 *      & b_i &  a_i  &       &      &        &        \\
        &     &       &       &      &        &         \\
 *      &\Box &  a_i  &       &      &        &         \\
        &     &       &\multicolumn{4}{c}%
          {\raisebox{2ex}[0pt]{\LARGE $\mathbf{A_2}$}}     \\
 b_i    &  *  &\Box   &       &      &        &        \\
        &     &       &       &      &        &         \\
*       &b_i  &\Box   &       &      &        &        \\
        &     &       &       &      &        &     \\
        &     &       &       &      &        &
\end{array}\right)
\end{array}
\end{equation*}
Since both size of $\{a_0,\cdots,a_{K-4}\}$ and $\{b_0,\cdots,b_{K-4}\}$ are $K-3$, we have
\begin{equation}
\label{eqs1}
h_1+h_2+h_3=K-3,\ \ \ \ \ \ \ \
h_2+h_4+h_5=K-3.
\end{equation}
Since the number of $``*"$s in each column is $Z=2F/K$, we have
\begin{equation}
\label{eqs3}
2+h_2+h_3+h_5 \leq \frac{2F}{K},\ \ \ 1+h_4 \leq \frac{2F}{K}\ \ \ \ 1+h_1 \leq \frac{2F}{K}.
\end{equation}
From \eqref{eqs1}, we have
\begin{equation}
\label{eqs6}
h_2=K-3-h_1-h_3.
\end{equation}
Put \eqref{eqs6} into the first inequality of \eqref{eqs3}, we have
\begin{equation}
\label{eqs7}
h_1 \geq K-1+h_5-\frac{2F}{K} \geq K-1-\frac{2F}{K}.
\end{equation}
On the other hand, from the third inequality of \eqref{eqs3}, we have
\begin{equation}
\label{eqs8}
h_1 \leq \frac{2F}{K}-1
\end{equation}
Combining \eqref{eqs7} and \eqref{eqs8}, we have
\begin{equation*}
K-1-\frac{2F}{K} \leq \frac{2F}{K}-1,
\end{equation*}
i.e.,
\begin{equation*}
F \geq \frac{K^2}{4}.
\end{equation*}
\end{proof}

From Lemma \ref{le-opg2} we know that the lower bound of $F$ can be achieved only if $K$ is even.
\begin{lemma}
\label{opgK-2}
If $\mathbf{P}$ is a $(K,1,F,Z,S)$ DPDA with $K\geq 3$, $\frac{Z}{F}=\frac{K-2}{K}$ and $R=\frac{F}{Z}-1$, then
\begin{equation*}
F \geq \left\{\begin{array}{cc}
K(K-2)  & \hbox{K is odd}\\
\frac{K(K-2)}{2} & \mbox{otherwise.}
\end{array}
\right.
\end{equation*}

\end{lemma}

\begin{proof}
Since $R=\frac{F}{Z}-1$, for any $k\in [0 , K)$, the number of
integers, whose superscript is $(k)$, is
\begin{equation}
\label{2btimes}
m_k=\frac{F}{K}\left(\frac{F}{Z}-1\right)=\frac{2F}{K(K-2)}
\end{equation}
by Theorem \ref{userbroadcast}. Since $m_k \geq 1$, we have
\begin{equation*}
\frac{2F}{K(K-2)}\geq 1.
\end{equation*}
That is,
\begin{equation*}
 F \geq \frac{K(K-2)}{2}.
\end{equation*}
If $K$ is odd, $K(K-2)/2$ is not an integer, then we have
\begin{equation*}
 F > \frac{K(K-2)}{2}.
\end{equation*}
This implies that $m_k=\frac{2F}{K(K-2)} \geq 2$. So we have $F \geq K(K-2)$.
\end{proof}

\section{New constructions}
\label{New cosntructions}
In this section, we will give two classes of DPDA with $Z/F=i/K$ for $i=2$, $K-2$ respectively.

\subsection{$Z/F=2/K$}
Let $q\geq 2$ be an integer. Given an integer $i\in[0,q^2)$ with $i=i_1q+i_0$ for integers $i_0$, $i_1\in[0,q)$, we refer to $i=(i_1,i_0)$ as the q-ary representation of $i$. Now we can define placement sets of $[0,2q)$ as
\begin{equation*}
V_{k_1,k_0}=\{(i_1,i_0)|i_{k_1}=k_0\}, k_0\in[0,q),\ \  k_1=0,\  \ 1.
\end{equation*}

\begin{construction}
\label{construction1}
For any $i=(i_1,i_0)\in[0,q^2)$ and $k=(k_1,k_0)\in[0,2q)$, define the entry in row $i$ and column $k$ in a $q^2 \times 2q$ array $\mathbf{P}=(p_{i,k})$ as follows: if $(i_1,i_0)\in V_{k_1,k_0}$, i.e. $i_{k_1}=k_0$, let $p_{i,k}=*$; if $(i_1,i_0)\notin V_{k_1,k_0}$, i.e. $i_{k_1}\neq k_0$, let $p_{i,k}$ be a non $``*"$ entry as follows:

\begin {equation*}
p_{i,k}=\begin{cases}
*,  & i_{k_1}=k_0, \\
\left((0,i_1),\{i_0,k_0\}\right)^{(q+i_1)},  & k_1=0,i_0\neq k_0,\\
\left((1,i_0),\{i_1,k_0\}\right)^{(i_0)},  & k_1=1,i_1\neq k_0.
\end{cases}
\end{equation*}
\end{construction}

\begin{example}
\label{examq=3}
When $q=3$, the array got from Construction \ref{construction1} is as follows:
\begin{eqnarray*}
\begin{small}
\mathbf{P} =\left(\begin{array}{cccccc}
 *                    & ((0,0),\{0,1\})^{(3)} & ((0,0),\{0,2\})^{(3)} &  *                    & ((1,0),\{0,1\})^{(0)} & ((1,0),\{0,2\})^{(0)}\\
((0,0),\{0,1\})^{(3)} &  *                    & ((0,0),\{1,2\})^{(3)} &  *                    & ((1,1),\{0,1\})^{(1)} & ((1,1),\{0,2\})^{(1)} \\
((0,0),\{0,2\})^{(3)} & ((0,0),\{1,2\})^{(3)} &  *                    &  *                    & ((1,2),\{0,1\})^{(2)} & ((1,2),\{0,2\})^{(2)}\\
  *                   & ((0,1),\{0,1\})^{(4)} & ((0,1),\{0,2\})^{(4)} & ((1,0),\{0,1\})^{(0)} &  *                    & ((1,0),\{1,2\})^{(0)}\\
((0,1),\{0,1\})^{(4)} &  *                    & ((0,1),\{1,2\})^{(4)} & ((1,1),\{0,1\})^{(1)} &  *                    & ((1,1),\{1,2\})^{(1)} \\
((0,1),\{0,2\})^{(4)} & ((0,1),\{1,2\})^{(4)} &  *                    & ((1,2),\{0,1\})^{(2)} &  *                    & ((1,2),\{1,2\})^{(2)} \\
  *                   & ((0,2),\{0,1\})^{(5)} & ((0,2),\{0,2\})^{(5)} & ((1,0),\{0,2\})^{(0)} & ((1,0),\{1,2\})^{(0)} &  *       \\
((0,2),\{0,1\})^{(5)} &  *                    & ((0,2),\{1,2\})^{(5)} & ((1,1),\{0,2\})^{(1)} & ((1,1),\{1,2\})^{(1)} &  *      \\
((0,2),\{0,2\})^{(5)} & ((0,2),\{1,2\})^{(5)} &  *                    & ((1,2),\{0,2\})^{(2)} & ((1,2),\{1,2\})^{(2)} &  *
\end{array}
\right)
\end{small}
\end{eqnarray*}


It's not difficult to check that above $\mathbf{P}$ satisfies C1 with $Z=3$, C2 with $S=18$, C3 and C4. So it is a $(6,1,9,3,18)$ DPDA. Its rate is $R=S/F=18/9=2=F/Z-1$, and $F=9=K^2/4$. So it reaches the lower bounds on $R$ and $F$ of DPDA by Theorem \ref{opdpda} and Lemma \ref{le-opg2}.
\end{example}

\begin{theorem}
\label{th-2/K}
For any integer $q\geq 2$, there exists a $(2q,1,q^2,q,q^3-q^2)$ DPDA.
\end{theorem}
\begin{proof}
Let $\mathbf{P}$ be the $q^2 \times 2q$  array got from Construction \ref{construction1}. First, it's easy to work out that there are $Z={q \choose 1}=q$ $``*"$s in each column of $\mathbf{P}$. Second, for any $i=(i_1,i_0)\in[0,q^2)$ and $k=(k_1,k_0)\in[0,2q)$, if $i_{k_1} \neq k_0$, $p_{i,k}$ is a non $``*"$ entry, for which there are two cases:
\begin{itemize}
  \item If $k_1=0$, $i_{k_1}=i_0\neq k_0$, then we have $(i_1,i_0)\in V_{0,i_0}\cap V_{1,i_1}$ and $(i_1,k_0)\in V_{0,k_0}\cap V_{1,i_1}$. So the subarray formed by rows $(i_1,i_0)$, $(i_1,k_0)$ and columns $(0,k_0)$, $(0,i_0)$, $(1,i_1)$ is the following form
      \begin{equation*}
      \bordermatrix{%
                 &(0,k_0) & (0,i_0) &(1,i_1) \cr
      (i_1,i_0)  & \Box   &    *    &   *    \cr
      (i_1,k_0)  &   *    &  \Box   &   *  }
      \end{equation*}
      with respect to row/column permutation. Hence the super combination $\mathcal{T}=\left((0,i_1),\{i_0,k_0\}\right)$ is used to denote the non $``*"$ entry $\Box$, and the signal is broadcasted by user $(1,i_1)=q+i_1$.

  \item If $k_1=1$, $i_{k_1}=i_1\neq k_0$, then we have $(i_1,i_0)\in V_{0,i_0}\cap V_{1,i_1}$ and $(k_0,i_0)\in V_{0,i_0}\cap V_{1,k_0}$. So the subarray formed by rows $(i_1,i_0)$, $(k_0,i_0)$ and columns $(1,k_0)$, $(1,i_1)$, $(0,i_0)$ is the following form
      \begin{equation*}
      \bordermatrix{%
                 &(1,k_0) & (1,i_1) &(0,i_0) \cr
      (i_1,i_0)  & \Box   &    *    &   *    \cr
      (k_0,i_0)  &   *    &  \Box   &   *  }
      \end{equation*}
      with respect to row/column permutation. Hence the super combination $\mathcal{T}=\left((1,i_0),\{i_1,k_0\}\right)$ is used to denote the non $``*"$ entry $\Box$, and the signal is broadcasted by user $(0,i_0)=i_0$.
\end{itemize}

Since $i_0,i_1,k_0$ traverse the set $[0,q)$, there are ${q \choose 2}\times q \times 2=q^3-q^2$ different super combinations in $\mathbf{P}$.
Then $\mathbf{P}$ satisfies C1 with $Z=q$, C2 with $S=q^3-q^2$, C3 and C4. So it is a $(2q,1,q^2,q,q^3-q^2)$ DPDA.
\end{proof}

Clearly the DPDA in Theorem \ref{th-2/K} reaches both the lower bounds on $R=S/F=q-1=F/Z-1$ and $F=q^2=K^2/4$ by Theorem \ref{opdpda} and Lemma \ref{le-opg2} respectively. From Table \ref{tab-known}, JCM scheme has $R=q-1$ and $F=2{K\choose 2}=4q^2-2q$ for the same parameters $K=2q$ and $M/N=1/q$. So our $F$ is about $1/4$ times as small as that of JCM scheme when $q$ is large enough.

\subsection{$Z/F=1-2/K$}

\subsubsection{$K$ is even}

Let $K\geq 4$ be even. If $K=4$, it easy to see that
\begin{eqnarray}
\label{P4}
\mathbf{P}_4 =\left(\begin{array}{cccc}
2^{(2)} &     *    &   *    & 1^{(1)}\\
*       & 2^{(2)}  &   *    & 0^{(0)}\\
3^{(3)} &   *      & 1^{(1)}&   * \\
  *     & 3^{(3)}  & 0^{(0)}&   *
\end{array}
\right)
\end{eqnarray}
is a $(4,1,4,2,4)$ DPDA. Furthermore, if $K=6$,
\begin{eqnarray*}
\mathbf{P}_6 =\left(\begin{array}{cccc|c|c}
2^{(2)} &     *    &   *    & 1^{(1)} &  *     &  *\\
*       & 2^{(2)}  &   *    & 0^{(0)} &  *     &  *\\
3^{(3)} &   *      & 1^{(1)}&   *     &  *     &  *\\
  *     & 3^{(3)}  & 0^{(0)}&   *     &  *     &  *\\   \hline
4^{(4)} &   *      &  *     &   *     &  *     &1^{(1)}\\
*       & 4^{(4)}  &  *     &   *     &  *     &0^{(0)}\\
*       &    *     &4^{(4)} &   *     &  *     &3^{(3)}\\
*       &    *     &  *     &4^{(4)}  &  *     &2^{(2)} \\  \hline
5^{(5)} &   *      &  *     &   *     &1^{(1)} &  *     \\
*       & 5^{(5)}  &  *     &   *     &0^{(0)} &  *     \\
*       &    *     &5^{(5)} &   *     &3^{(3)} &  *\\
*       &    *     &  *     &5^{(5)}  &2^{(2)} &  *
\end{array}
\right)
\end{eqnarray*}
is a $(6,1,12,8,6)$ DPDA. It's easy to check that both $\mathbf{P}_4$ and $\mathbf{P}_6$ reach the lower bounds on $R$ and $F$ by Theorem \ref{opdpda} and Lemma \ref{opgK-2} respectively. In fact, $\mathbf{P}_6$ can be written as follows
\begin{eqnarray*}
\mathbf{P}_6 =\left(\begin{array}{ccc}
\mathbf{P}_4&        *               &           *                 \\
4^{(4)}\mathbf{I}_4 &        *               &  \alpha{_4^ \mathrm{ T }}  \\
5^{(5)}\mathbf{I}_4 &\alpha{_4^ \mathrm{ T }}&              *
\end{array}
\right),
\end{eqnarray*}
where $\mathbf{P}_4$ is defined by \eqref{P4}, $\mathbf{I}_n$ denotes the $n \times n$ array
\begin{eqnarray}
\label{In}
\mathbf{I}_n=\left(\begin{array}{cccc}
1      &  *    &  \cdots &  *\\
*      &  1    &  \cdots &  *\\
\vdots &\vdots &  \ddots &  \vdots\\
*      &  *    &  \cdots & 1
\end{array}
\right),
\end{eqnarray}
 $\alpha_4$ is a $4$-dimensional vector
\begin{equation}
\label{alpha4}
\alpha_4=[1^{(1)},0^{(0)},3^{(3)},2^{(2)}].
\end{equation}
More generally, we propose a recursive construction as follows:

\begin{construction}
\label{construction2}
Let $n\geq 4$ be even, we define a $\frac{n(n+2)}{2} \times (n+2)$ array
\begin{eqnarray}
\label{Peven}
\mathbf{P}_{n+2}=\left(\begin{array}{ccc}
\mathbf{P}_n&        *               &           *                 \\
n^{(n)}\mathbf{I}_n &        *               &  \alpha{_n^ \mathrm{ T }}  \\
(n+1)^{(n+1)}\mathbf{I}_n &\alpha{_n^ \mathrm{ T }}&              *
\end{array}
\right),
\end{eqnarray}
where $\mathbf{P}_4$, $\mathbf{I}_n$ and $\alpha_4$ are defined by\eqref{P4}, \eqref{In} and \eqref{alpha4} respectively, and
\begin{equation}
\label{alphan}
\alpha_n=[\alpha_{n-2},(n-1)^{(n-1)},(n-2)^{(n-2)}].
\end{equation}
\end{construction}

\begin{theorem}
\label{th-K_EVEN_K-2/K}
For any integer $q\geq 2$, there exists a $(2q,1,2q(q-1),2(q-1)^2,2q)$ DPDA.
\end{theorem}

The proof of Theorem \ref{th-K_EVEN_K-2/K} is included in Appendix B. Clearly the DPDA in Theorem \ref{th-K_EVEN_K-2/K} reaches both the lower bounds on $R=S/F=1/(q-1)=F/Z-1$ and $F=2q(q-1)=K(K-2)/2$ by Theorem \ref{opdpda} and Lemma \ref{opgK-2} respectively. From Table \ref{tab-known}, JCM scheme has $R=\frac{1}{q-1}$ and $F=2q(q-1)(2q-1)$ for the same parameters $K=2q$ and $M/N=(q-1)/q$. This implies that our $F$ is about $1/(2q-1)$ times as small as that of JCM scheme.

\subsubsection{$K$ is odd}
Let $K\geq 3$ be odd. If $K=3$, it is easy to see that
\begin{eqnarray}
\label{P3}
\mathbf{P}_3 =\left(\begin{array}{ccc}
   *    &  0^{(0)} & 1^{(0)}\\
3^{(1)} &   *      & 2^{(1)}\\
4^{(2)} & 5^{(2)}  &  *
\end{array}
\right)
\end{eqnarray}
is a $(3,1,3,1,6)$ DPDA. Furthermore, if $K=5$,
\begin{eqnarray*}
\mathbf{P}_5 =\left(\begin{array}{ccc|c|c}
   *    &  0^{(0)} & 1^{(0)} &   *    &  *  \\
3^{(1)} &   *      & 2^{(1)} &   *    &  *  \\
4^{(2)} & 5^{(2)}  &  *      &   *    &  *  \\   \hline
6^{(3)} &   *      &  *      &   *    & 2^{(1)}\\
*       & 6^{(3)}  &  *      &   *    & 4^{(2)}\\
*       &    *     &6^{(3)}  &   *    & 0^{(0)}\\ \hline
8^{(4)} &   *      &  *     & 2^{(1)} &  *     \\
*       & 8^{(4)}  &  *     & 4^{(2)} &  *     \\
*       &    *     &8^{(4)} & 0^{(0)} &  *  \\ \hline
7^{(3)} &   *      &  *      &   *    & 5^{(2)}\\
*       & 7^{(3)}  &  *      &   *    & 1^{(0)}\\
*       &    *     &7^{(3)}  &   *    & 3^{(1)}\\ \hline
9^{(4)} &   *      &  *     & 5^{(2)} &  *     \\
*       & 9^{(4)}  &  *     & 1^{(0)} &  *     \\
*       &    *     &9^{(4)} & 3^{(1)} &  *
\end{array}
\right)
\end{eqnarray*}
is a $(5,1,15,9,10)$ DPDA. Both $\mathbf{P}_3$  and $\mathbf{P}_5$ reach the lower bounds on $R$ and $F$ by Theorem \ref{opdpda} and Lemma \ref{opgK-2} respectively. In fact, $\mathbf{P}_5$ can be written as follows
\begin{eqnarray*}
\mathbf{P}_5=\left(\begin{array}{ccc}
\mathbf{P}_3&        *               &           *                 \\
6^{(3)}\mathbf{I}_3 &        *               &  \beta{_3^ \mathrm{ T }}  \\
8^{(4)}\mathbf{I}_3 &\beta{_3^ \mathrm{ T }}&              *    \\
7^{(3)}\mathbf{I}_3 &        *               &  \gamma{_3^ \mathrm{ T }}  \\
9^{(4)}\mathbf{I}_3 &\gamma{_3^ \mathrm{ T }}&              *
\end{array}
\right),
\end{eqnarray*}
where $\beta_3$ and $\gamma_3$ are $3$-dimensional vectors
\begin{equation}
\label{beta3}
\beta_3=[2^{(1)},4^{(2)},0^{(0)}],
\end{equation}
\begin{equation}
\label{gamma3}
\gamma_3=[5^{(2)},1^{(0)},3^{(1)}].
\end{equation}
More generally, we propose a recursive construction as follows:

\begin{construction}
\label{construction3}
Let $n\geq 3$ be odd, we define a $n(n+2) \times (n+2)$ array
\begin{eqnarray}
\label{Podd}
\mathbf{P}_{n+2}=\left(\begin{array}{ccc}
\mathbf{P}_n               &        *               &           *                 \\
(2n)^{(n)}\mathbf{I}_n     &        *               &  \beta{_n^ \mathrm{ T }}  \\
(2n+2)^{(n+1)}\mathbf{I}_n & \beta{_n^ \mathrm{ T }}&              *              \\
(2n+1)^{(n)}\mathbf{I}_n   &        *               &  \gamma{_n^ \mathrm{ T }}  \\
(2n+3)^{(n+1)}\mathbf{I}_n &\gamma{_n^ \mathrm{ T }}&              *              \\
\end{array}
\right),
\end{eqnarray}
where $\mathbf{P}_3$, $\beta_3$ and $\gamma_3$ are defined by \eqref{P3}, \eqref{beta3} and \eqref{gamma3} respectively, and
\begin{equation}
\label{betan}
\beta_n=[\beta_{n-2},(2n-2)^{(n-1)},(2n-4)^{(n-2)}],
\end{equation}
\begin{equation}
\label{gamman}
\gamma_n=[\gamma_{n-2},(2n-1)^{(n-1)},(2n-3)^{(n-2)}].
\end{equation}
\end{construction}

\begin{theorem}
\label{th-K_ODD_K-2/K}
For any integer $q\geq 1$, there exists a $(2q+1, 1,4q^2-1, (2q-1)^2, 4q+2)$ DPDA.
\end{theorem}
The proof of Theorem \ref{th-K_ODD_K-2/K} is included in Appendix C. Clearly the DPDA in Theorem \ref{th-K_ODD_K-2/K} reaches both the lower bounds on $R=2/(2q-1)=F/Z-1$ and $F=4q^2-1=K(K-2)$ by Theorem \ref{opdpda} and Lemma \ref{le-opg2} respectively. From Table \ref{tab-known}, JCM scheme has $R=2/(2q-1)$ and $F=q(4q^2-1)$ for the same parameters $K=2q+1$ and $M/N=1-2/(2q+1)$. So our $F$ is about $1/q$ times as small as that of JCM scheme.
\section{Conclusions}
\label{conclusions}
In this paper, we focused on coded caching schemes for D2D network with $R$ and $F$ as small as possible . Firstly we proposed a $(K,L',F,Z,S)$ DPDA to characterize the two phases of an $F$-division $(K,M,N,L,L')$ coded caching scheme. As a result, the problem of designing a determined coded caching scheme for D2D network could be transformed into the problem of designing an appropriate DPDA. In fact, JCM scheme can be represented by a special class of DPDA. Secondly the lower bound $R\geq N/M-1$ of a DPDA was derived. This lower bound was exactly the rate of JCM scheme. Furthermore, constructing a $(K,L',F,Z,S)$ DPDA could be reduced to constructing a $(K,1,F,Z,S')$ DPDA with the same $R=N/M-1$. Thirdly based on this investigation, we derived the lower bound on $F$ of a DPDA with $R=N/M-1$ and some $N/M$. Consequently we showed that the $F$ in JCM scheme could not achieve our lower bound when $M/N=2/K,(K-2)/K$.
 Finally we constructed three classes of DPDA which achieved both the lower bounds on $R$ and $F$ when $M/N=2/K,(K-2)/K$.

It is interesting to show whether the $F$ in JCM scheme reaches the lower bound or not for other values of $M/N$. And it would be meaningful to characterize the tradeoff between $R$ and $F$ for other parameters $K$ and $M/N$ and construct the related Pareto-optimal DPDA.

\section*{Appendix A:Proof of Theorem \ref{JCMtoDPDA}}
First the following notations are very useful.

Let $H=\{h_0,h_1,\cdots,h_{n-1}\}$ be an ordered set, $h_0<h_1<\cdots<h_{n-1}$ are integers,
\begin{itemize}
\item Let $H|_m$ denotes the $m$th element of $H$, i.e. $H|_m=h_m$ for any $m\in[0,n)$.
\item Let us arrange all $g(g\in[1,n])$ sized subsets of $H$ in the lexicographic order, then for any $g$ sized subset $G$ of $H$,  define $f_{H,g}(G)$ to be its order minus $1$. Clearly, $f_{H,g}$ is a bijection from $\{G|G\subset H,|G|=g\}$ to $[0,{n\choose g})$. For example, when $H=\{2,5,8,10\}$, all $2$ sized subsets are ordered as $\{2,5\},\{2,8\},\{2,10\},\{5,8\},\{5,10\}$ and $\{8,10\}$. Accordingly, $f_{H,2}(\{2,5\})=0$, $f_{H,2}(\{2,8\})=1$, $f_{H,2}(\{2,10\})=2$, $f_{H,2}(\{5,8\})=3$, $f_{H,2}(\{5,10\})=4$, $f_{H,2}(\{8,10\})=5$.
\end{itemize}

Now let us propose the proof of Theorem \ref{JCMtoDPDA}.
\begin{proof}
For each $t=KM/N \in [1,K)$, the placement and delivery phase in JCM scheme are as follows
\begin{itemize}
\item {\bf Placement Phase:} Each block of each file is split into $F=t{K\choose t}$ packets indexed by $\{(T,j)|T\subset [0,K),|T|=t, j\in [0,t)\}$, such that for any $l\in[0,L)$ and $f\in[0,N)$, the $l$-th block of the file $W_f$ is indicated by
    \begin{equation*}
    W{_f^l}=\{W{_{f}^{l,T,j}|T\subset [0,K),|T|=t, j\in [0,t)}\}.
    \end{equation*}
    And user $k$ caches all the packets such that $k\in T$, i.e.,
    \begin{equation*}
    \mathcal{Z}_k=\{W{_{f}^{l,T,j}| k\in T,f\in [0,N),j\in [0,t)},l\in[0,L)\}.
    \end{equation*}
\item {\bf Delivery Phase:} Assume that the request vectors are ${\bf d}=(d_0,d_1,\cdots,d_{K-1})$ and ${\bf b}=(b_0,b_1,\cdots, b_{K-1})$. For any $(t+1)$-sized subset of $[0,K)$ indicated by $U$, let $U=\{u_0,u_1,\cdots, u_t\}$, user $u_k\in U$ sends
    \begin{equation*}
    \underset{u_m\in U\setminus \{u_k\}}{\bigoplus}  W{_{d_{u_m}}^{b_{u_m}+l,U\setminus \{u_m\},j_{k,m}}},
    \end{equation*}
    where $l\in[0,L')$, $j_{k,m}\in[0,t)$ and for any $m\in[0,t]$, $\{j_{k,m}|k\in [0,t]\setminus\{m\}\}=[0,t)$, which means that the packets $W{_{d_{u_m}}^{b_{u_m}+l, U\setminus \{u_m\},0}},W{_{d_{u_m}}^{b_{u_m}+l,U\setminus \{u_m\},1}},\cdots, W{_{d_{u_m}}^{b_{u_m}+l,U\setminus \{u_m\},t-1}}$, which are needed by user $u_m$, are provided by users in $U\setminus \{u_m\}$, and user $u_k\in U\setminus \{u_m\}$ provides $W{_{d_{u_m}}^{b_{u_m}+l,U\setminus \{u_m\},j_{k,m}}}$.
\end{itemize}
It is worth to noting that the delivery phase in JCM scheme is not unique.
Now we will show that a $(K,M,N,L,1)$ JCM scheme is corresponding to a DPDA constructed as follows:

Let $\mathbf{Q}$ be an $t{K\choose t} \times K$ array. Denote its rows by $\{(T,j)|T\subset \mathcal{K}, |T|=t,j \in [0,t)\}$ and columns by $0,1,\cdots,K-1$, respectively. Then, define the entry in row $(T,j)$ and column $k$ as follows
\begin{equation}
\label{eqJiDPDA}
Q((T,j),k)=
\begin{cases}
\ast  & k \in T, \\
T\cup\{k\}^{(m)}  & k \notin T,
\end{cases}
\end{equation}
where
\begin{equation}
\label{eqm}
m=
\begin{cases}
T\cup \{k\}|_j,  & 0\leq j \leq t-f_{T\cup \{k\},t}(T)-1, \\
T\cup \{k\}|_{j+1},  & t-f_{T\cup \{k\},t}(T) \leq j \leq t-1.
\end{cases}
\end{equation}

Next, we prove that the array $\mathbf{Q}$ defined in \eqref{eqJiDPDA} is a $(K,1,t{K\choose t}, t{K-1\choose t-1},(t+1){K\choose t+1})$ DPDA, i.e. we should consider the following conditions
\begin{itemize}
\item It's easy to see that there are $Z=t{K-1\choose t-1}$ $``*"$s in each column of $\mathbf{Q}$. C$1$ holds.
\item The number of $(t+1)$-sized subsets of $\mathcal{K}$ is ${K\choose t+1}$, and for any such subset, each user in it broadcasts a multicast coded message. So $S=(t+1){K\choose t+1}$, C$2$ holds.
\item If $Q((T,j),k)=T\cup\{k\}^{(m)}$, then  $m\in T$ yielding $Q((T,j),m)=*$ by \eqref{eqJiDPDA}. Otherwise if $m\notin T$, then $m=k$ by \eqref{eqm} and we only need to consider the following two subcases.
    \begin{itemize}
    \item When $0\leq j \leq t-f_{T\cup \{k\},t}(T)-1$, we have $k=T\cup \{k\}|_j$ by \eqref{eqm}. This implies $f_{T\cup \{k\},t}(T)=t-j$, i.e., $j \leq t-f_{T\cup \{k\},t}(T)-1=j-1$, a contradiction to our hypothesis.
    \item When $t-f_{T\cup \{k\},t}(T) \leq j \leq t-1$, we have $k=T\cup \{k\}|_{j+1}$ by \eqref{eqm}. This implies $f_{T\cup \{k\},t}(T)=t-j-1$, i.e., $j+1=t-f_{T\cup \{k\},t}(T) \leq j$, a contradiction to our hypothesis.
    \end{itemize}
    C$3$ holds.
\item For any two distinct entries $Q((T_1,j_1),k_1)$ and $Q((T_2,j_2),k_2)$, if $Q((T_1,j_1),k_1)=Q((T_2,j_2),k_2)=T_1\cup\{k_1\}^{(m)}$, then $k_1\notin T_1$, $k_2\notin T_2$ and $T_1\cup \{k_1\}=T_2\cup \{k_2\}$ by \eqref{eqJiDPDA}. Since $Q((T_1,j_1),k_1)$ and $Q((T_2,j_2),k_2)$ are distinct entries, we have $k_1\neq k_2$. Otherwise if $k_1=k_2$, we have $T_1=T_2$, then $j_1\neq j_2$, without loss of generality, assume that $j_1< j_2$. Since $Q((T_1,j_1),k_1)=Q((T_2,j_2),k_2)=T_1\cup\{k_1\}^{(m)}$, we have $j_1\in [0,t-f_{T\cup \{k\},t}(T)-1]$ and $j_2\in [t-f_{T\cup \{k\},t}(T),t-1]$
    by \eqref{eqm}. As a result, $m=T_1\cup \{k_1\}|_{j_1}=T_1\cup \{k_1\}|_{j_2+1}$, which is a contradiction. Hence, we have $k_1\in T_2$ and $k_2\in T_1$, which implies that $Q((T_1,j_1),k_2)=Q((T_2,j_2),k_1)=*$. C$4$ holds.
\end{itemize}
So the array $\mathbf{Q}$ corresponding to the $(K,M,N,L,1)$ JCM scheme is a $(K,1,F, Z,S)$ DPDA, where $F=t{K\choose t}, Z=t{K-1\choose t-1},S=(t+1){K\choose t+1}$.

More generally, the $(K,M,N,L,L')$ JCM scheme can be represented by the following $(K,L',F,Z,L'S)$ DPDA.
\begin{eqnarray}
\left(\begin{array}{c}
\mathbf{Q}     \\
\mathbf{Q}+S\cdot {\bf J}_{F\times K}   \\
\mathbf{Q}+2S\cdot {\bf J}_{F\times K}  \\
\vdots         \\
\mathbf{Q}+(L'-1)S\cdot {\bf J}_{F\times K}
\end{array}
\right)
\end{eqnarray}
where ${\bf J}_{F\times K}$ is an $F\times K$ array with each entry $1$.
\end{proof}
\begin{example}
\label{examJM}
Consider the $(4,2,4,2,1)$ JCM scheme, yielding $t=2$. Each block is divided into $t{K\choose t}=12$ packets with the following labeling: for $f=0,1,2,3$ and $l=0,1$, denote the $l$-th block of the file $W_f$ by $$W{_f^l}=\{W{_{f}^{l,\{0,1\},j}},W{_{f}^{l,\{0,2\},j}},W{_{f}^{l,\{0,3\},j}},W{_{f}^{l,\{1,2\},j}},W{_{f}^{l,\{1,3\},j}},W{_{f}^{l,\{2,3\},j}}|j=0,1\}$$.The caches are given by
\begin{align*}
    \mathcal{Z}_0 &= \{W{_{f}^{l,\{0,1\},j}},W{_{f}^{l,\{0,2\},j}},W{_{f}^{l,\{0,3\},j}} |\ f\in[0,4),l\in[0,2),j\in[0,2)\}\\
    \mathcal{Z}_1 &=  \{W{_{f}^{l,\{0,1\},j}},W{_{f}^{l,\{1,2\},j}},W{_{f}^{l,\{1,3\},j}} |\ f\in[0,4),l\in[0,2),j\in[0,2)\}\\
    \mathcal{Z}_2 &=  \{W{_{f}^{l,\{0,2\},j}},W{_{f}^{l,\{1,2\},j}},W{_{f}^{l,\{2,3\},j}} |\ f\in[0,4),l\in[0,2),j\in[0,2)\}\\
    \mathcal{Z}_3 &=  \{W{_{f}^{l,\{0,3\},j}},W{_{f}^{l,\{1,3\},j}},W{_{f}^{l,\{2,3\},j}} |\ f\in[0,4),l\in[0,2),j\in[0,2)\}
\end{align*}
Assuming that the request vectors are $\textbf{d}=(0,1,2,3)$ and $\textbf{b}=(0,1,0,1)$. The following multicast coded messages are sent:
\begin{align*}
    X_{0,\textbf{d},\textbf{b}} &=W{_{1}^{1,\{0,2\},0}} \oplus W{_{2}^{0,\{0,1\},0}},\quad  W{_{1}^{1,\{0,3\},0}} \oplus W{_{3}^{1,\{0,1\},0}},\quad W{_{2}^{0,\{0,3\},0}} \oplus W{_{3}^{1,\{0,2\},0}}\\
    X_{1,\textbf{d},\textbf{b}} &=W{_{0}^{0,\{1,2\},0}} \oplus W{_{2}^{0,\{0,1\},1}},\quad W{_{0}^{0,\{1,3\},0}} \oplus W{_{3}^{1,\{0,1\},1}},\quad W{_{2}^{0,\{1,3\},0}} \oplus W{_{3}^{1,\{1,2\},0}}\\
    X_{2,\textbf{d},\textbf{b}} &=W{_{0}^{0,\{1,2\},1}} \oplus W{_{1}^{1,\{0,2\},1}},\quad W{_{0}^{0,\{2,3\},0}} \oplus W{_{3}^{1,\{0,2\},1}},\quad W{_{1}^{1,\{2,3\},0}} \oplus W{_{3}^{1,\{1,2\},1}}\\
    X_{3,\textbf{d},\textbf{b}} &=W{_{0}^{0,\{1,3\},1}} \oplus W{_{1}^{1,\{0,3\},1}},\quad W{_{0}^{0,\{2,3\},1}} \oplus W{_{2}^{0,\{0,3\},1}},\quad W{_{1}^{1,\{2,3\},1}} \oplus W{_{2}^{0,\{1,3\},1}}
\end{align*}

It's easy to check that the JCM scheme described in example \ref{examJM} is corresponding to the following DPDA
\begin{equation}
\label{(4,1,12,6,12)DPDA}
\mathbf{Q}=\bordermatrix{%
            &       0       &        1       &        2        &        3       \cr
(\{0,1\},0) &       *       &        *       & \{0,1,2\}^{(0)} &\{0,1,3\}^{(0)} \cr
(\{0,2\},0) &       *       &\{0,1,2\}^{(0)} &        *        &\{0,2,3\}^{(0)} \cr
(\{0,3\},0) &       *       &\{0,1,3\}^{(0)} & \{0,2,3\}^{(0)} &        *       \cr
(\{1,2\},0) &\{0,1,2\}^{(1)}&        *       &        *        &\{1,2,3\}^{(1)} \cr
(\{1,3\},0) &\{0,1,3\}^{(1)}&        *       & \{1,2,3\}^{(1)} &        *       \cr
(\{2,3\},0) &\{0,2,3\}^{(2)}&\{1,2,3\}^{(2)} &        *        &        *       \cr
(\{0,1\},1) &       *       &        *       & \{0,1,2\}^{(1)} &\{0,1,3\}^{(1)} \cr
(\{0,2\},1) &       *       &\{0,1,2\}^{(2)} &        *        &\{0,2,3\}^{(2)} \cr
(\{0,3\},1) &       *       &\{0,1,3\}^{(3)} & \{0,2,3\}^{(3)} &        *       \cr
(\{1,2\},1) &\{0,1,2\}^{(2)}&        *       &        *        &\{1,2,3\}^{(2)} \cr
(\{1,3\},1) &\{0,1,3\}^{(3)}&        *       & \{1,2,3\}^{(3)} &        *       \cr
(\{2,3\},1) &\{0,2,3\}^{(3)}&\{1,2,3\}^{(3)} &        *        &        *       }
\end{equation}

\end{example}

\section*{Appendix B:Proof of Theorem \ref{th-K_EVEN_K-2/K}}
In order to prove Theorem \ref{th-K_EVEN_K-2/K}, the following statements are very useful.
\begin{proposition}
\label{proposition1}
For any even integer $K\geq 4$, let $\mathbf{P}_{K}$ be recursively got from \eqref{Peven}, then there are $K$ integers in $\mathbf{P}_{K}$, and for any integer $s \in [0,K)$, its superscript is $(s)$ and
\begin{itemize}
\item if $s$ is even, $s^{(s)}$ only appears in columns $[0,K)\setminus \{s,s+1\}$ ;
\item if $s$ is odd, $s^{(s)}$ only appears in columns $[0,K)\setminus \{s-1,s\}$.
\end{itemize}
\end{proposition}
\begin{proof}
we use mathematic induction. First, when $K=4$, it's easy to check that the proposition is true. Assume that the proposition is true when $K=n$. When $K=n+2$, there are two more integers $\{n^{(n)},(n+1)^{(n+1)}\}$ in $\mathbf{P}_{n+2}$ than in $\mathbf{P}_{n}$ by \eqref{Peven} and \eqref{alphan}. So there are $n+2$ integers in $\mathbf{P}_{n+2}$, and for any integer $s \in [0,n+2)$, its superscript is $(s)$, and
\begin{itemize}
\item if $s\in[0,n)$ is even, then $s^{(s)}$ only appears in columns $[0,n)\setminus \{s,s+1\}$ of $\mathbf{P}_{n}$ by hypothesis. In addition, $s^{(s)}$ appears in $\alpha_n$ by \eqref{alpha4} and \eqref{alphan}, so $s^{(s)}$ only appears in columns $[0,n+2)\setminus \{s,s+1\}$ of $\mathbf{P}_{n+2}$ by \eqref{Peven};
\item if $s\in[0,n)$ is odd, then $s^{(s)}$ only appears in columns $[0,n)\setminus \{s-1,s\}$ of $\mathbf{P}_{n}$ by hypothesis. In addition, $s^{(s)}$ appears in $\alpha_n$ by \eqref{alpha4} and \eqref{alphan}, so $s^{(s)}$ only appears in columns $[0,n+2)\setminus \{s-1,s\}$ of $\mathbf{P}_{n+2}$ by \eqref{Peven};
\item if $s\in \{n,n+1\}$, $s^{(s)}$ only appears in $s^{(s)}\mathbf{I}_n$ of $\mathbf{P}_{n+2}$ by \eqref{Peven}, so the even integer $n^{(n)}$ and the odd integer $(n+1)^{(n+1)}$ only appear in columns $[0,n)=[0,n+2)\setminus \{n,n+1\}$ of $\mathbf{P}_{n+2}$.
\end{itemize}
\end{proof}

\begin{proposition}
\label{proposition2}
For any even integer $K\geq 4$, let $\mathbf{P}_{K}$ be recursively got from \eqref{Peven}. If $K=n+2$, $k\in\{n,n+1\}$ and $\mathbf{P}_{n+2}(i,k)=s^{(s)}$, then $s\in[0,n)$ and
\begin{itemize}
\item if $s$ is odd, the entries in row $i$ and columns $[0,n)\setminus \{s-1\}$ of $\mathbf{P}_{n+2}$ are all $``*"$s;
\item if $s$ is even, the entries in row $i$ and columns $[0,n)\setminus \{s+1\}$ of $\mathbf{P}_{n+2}$ are all $``*"$s.
\end{itemize}
\end{proposition}
\begin{proof}
Since $\mathbf{P}_{n+2}(i,k)=s^{(s)}$ and $k\in\{n,n+1\}$, then from \eqref{Peven} we have $s\in[0,n)$, and
\begin{itemize}
\item if $s$ is odd, the $i$-th row and columns $[0,n)\cup \{k\}$ of $\mathbf{P}_{n+2}$ is as follows
\begin{equation*}
      \bordermatrix{%
                 & 0 & 1 &\cdots & s-2 & s-1        & s    & \cdots & n-1 & k       \cr
              i  & * & * &\cdots & *   &  s'^{(s')} & *    & \cdots & *   & s^{(s)} }
\end{equation*}
\item if $s$ is even, the $i$-th row and columns $[0,n)\cup \{k\}$ of $\mathbf{P}_{n+2}$ is as follows
\begin{equation*}
      \bordermatrix{%
                 & 0 & 1 &\cdots & s   & s+1        & s+2   & \cdots & n-1 & k      \cr
              i  & * & * &\cdots & *   &  s'^{(s')} & *     & \cdots & *   & s^{(s)} }
\end{equation*}
\end{itemize}
where $s'$ is some integer in $\{n,n+1\}$. So the proposition is true.
\end{proof}

Now let us propose the proof of Theorem \ref{th-K_EVEN_K-2/K}.
\begin{proof}
First we use mathematic induction to prove that $\mathbf{P}_K$ recursively got from \eqref{Peven} is a $(K,1,K(K-2)/2,(K-2)^2/2,K)$ DPDA. When $K=4$, $\mathbf{P}_4$ in \eqref{P4} is a $(4,1,4,2,4)$ DPDA. When $K=n$, assume that $\mathbf{P}_{n}$ is a $(n,1,n(n-2)/2,(n-2)^2/2,n)$ DPDA. When $K=n+2$, we should consider the following cases.
\begin{itemize}
\item $\mathbf{P}_{n+2}$ is an $\left(n(n-2)/2+2n\right) \times (n+2)$ array by \eqref{Peven}, so it has $F=n(n+2)/2$ rows and $K=n+2$ columns.
\item For any $k\in[0,n)$, there are $(n-2)^2/2+2(n-1)=n^2/2$ $``*"$s in column $k$ of $\mathbf{P}_{n+2}$ by \eqref{Peven}. For any $k\in \{n,n+1\}$, there are $n(n-2)/2+n=n^2/2$ $``*"$s in column $k$ of $\mathbf{P}_{n+2}$ by \eqref{Peven}. So there are $Z=n^2/2$ $``*"$s in each column of $\mathbf{P}_{n+2}$, i.e., C$1$ holds.
\item There are $S=n+2$ integers in $\mathbf{P}_{n+2}$ and for any integer $s\in[0:n+2)$, its superscript is $(s)$ by Proposition \ref{proposition1}. So C$2$ holds.
\item Let $\mathbf{P}_{n+2}(i,j)=s^{(s)}$,
    \begin{itemize}
    \item if $s,j\in[0,n)$, $\mathbf{P}_{n+2}(i,s)=\mathbf{P}_{n}(i,s)=*$ by \eqref{Peven} and hypothesis;
    \item if $s\in [0,n)$, $j\in\{n,n+1\}$, $\mathbf{P}_{n+2}(i,s)=*$ by Proposition \ref{proposition2};
    \item if $s\in\{n,n+1\}$, it's easy to see that $\mathbf{P}_{n+2}(i,s)=*$ by \eqref{Peven}.
    \end{itemize}
    So C$3$ holds.
\item For any two distinct entries $\mathbf{P}_{n+2}(i_1,j_1)$ and $\mathbf{P}_{n+2}(i_2,j_2)$, assume that $\mathbf{P}_{n+2}(i_1,j_1)=\mathbf{P}_{n+2}(i_2,j_2)=s^{(s)}$.
    \begin{itemize}
    \item When $s\in \{n,n+1\}$, $s^{(s)}$ only appears in $s^{(s)}\mathbf{I}_n$ of $\mathbf{P}_{n+2}$ by \eqref{Peven}. So $\mathbf{P}_{n+2}(i_2,j_1)=\mathbf{P}_{n+2}(i_1,j_2)=*$;
    \item When $s\in [0,n)$, let us consider the values of $j_1$ and $j_2$.
         \begin{itemize}
         \item[$1)$] If $j_1,j_2\in[0,n)$, we have $\mathbf{P}_{n+2}(i_2,j_1)=\mathbf{P}_{n}(i_2,j_1)=*$,$\mathbf{P}_{n+2}(i_1,j_2)=\mathbf{P}_{n}(i_1,j_2)=*$ by \eqref{Peven} and hypothesis;
         \item[$2)$] If $j_1,j_2\in \{n,n+1\}$, we have $\mathbf{P}_{n+2}(i_2,j_1)=\mathbf{P}_{n+2}(i_1,j_2)=*$ by \eqref{Peven};
         \item[$3)$] If one of $j_1,j_2$ is in $[0,n)$, and the other is in $\{n,n+1\}$, without loss of generality, let $j_1<j_2$, i.e., $j_1\in [0,n)$ and $j_2\in \{n,n+1\}$. Then $\mathbf{P}_{n+2}(i_1,j_1)$ is in $\mathbf{P}_{n}$ and $\mathbf{P}_{n+2}(i_2,j_2)$ is in $\alpha{_n^ \mathrm{ T }}$. Hence, we have $\mathbf{P}_{n+2}(i_1,j_2)=*$ by \eqref{Peven}.
             \begin{itemize}
             \item If $s$ is even, $s^{(s)}$ only appears in columns $[0,n)\setminus \{s,s+1\}$ of $\mathbf{P}_{n}$ by Proposition \ref{proposition1}, we have $j_1\neq s+1$. So we have $\mathbf{P}_{n+2}(i_2,j_1)=*$ by Proposition \ref{proposition2};
             \item If $s$ is odd, $s^{(s)}$ appears in columns $[0,n)\setminus \{s-1,s\}$ of $\mathbf{P}_{n}$ by Proposition \ref{proposition1}, we have $j_1\neq s-1$. So we have $\mathbf{P}_{n+2}(i_2,j_1)=*$ by Proposition \ref{proposition2}.
             \end{itemize}

        \end{itemize}

    \end{itemize}
 So C$4$ holds.
\end{itemize}

Hence, $\mathbf{P}_{n+2}$ is a $(n+2,1,n(n+2)/2,n^2/2,n+2)$ DPDA. So $\mathbf{P}_K$ recursively got from \eqref{Peven} is a $(K,1,K(K-2)/2,(K-2)^2/2,K)$ DPDA. Since $K\geq 4$ is even, let $K=2q, q\geq 2$, then $\mathbf{P}_K$ is a $(2q,1,2q(q-1),2(q-1)^2,2q)$ DPDA.
\end{proof}

\section*{Appendix C:Proof of Theorem \ref{th-K_ODD_K-2/K}}
In order to prove Theorem \ref{th-K_ODD_K-2/K}, the following statements are very useful.

\begin{proposition}
\label{proposition3}
For any odd integer $K\geq 3$, let $\mathbf{P}_K$ be recursively got from \eqref{Podd}, then there are $2K$ integers in $\mathbf{P}_K$, and for any integer $s\in[0,2K)$, the superscript is $(\lfloor \frac{s}{2} \rfloor)$, and the following statements hold.
\begin{itemize}
\item When $s=0$ or $5$, $s^{(\lfloor \frac{s}{2} \rfloor)}$ only appears in columns $[0,K)\setminus \{0,2\}$;
\item When $s=1$ or $2$, $s^{(\lfloor \frac{s}{2} \rfloor)}$ only appears in columns $[0,K)\setminus \{0,1\}$;
\item When $s=3$ or $4$, $s^{(\lfloor \frac{s}{2} \rfloor)}$ only appears in columns $[0,K)\setminus \{1,2\}$;
\item When $s\in[6,2K)$ and $\lfloor\frac{s}{2}\rfloor$ is even, $s^{(\lfloor\frac{s}{2}\rfloor)}$ only appears in columns $[0,K)\setminus \{\lfloor\frac{s}{2}\rfloor-1,\lfloor\frac{s}{2}\rfloor\}$;
\item When $s\in[6,2K)$ and $\lfloor\frac{s}{2}\rfloor$ is odd, $s^{(\lfloor\frac{s}{2}\rfloor)}$ only appears in columns $[0,K)\setminus \{\lfloor\frac{s}{2}\rfloor,\lfloor\frac{s}{2}\rfloor+1\}$.
\end{itemize}
\end{proposition}
\begin{proof}
we use mathematic induction. First, when $K=3$, it's easy to check that the proposition is true. Assume that the proposition is true when $K=n$. When $K=n+2$, there are four more integers $\{(2n-4)^{(n-2)},(2n-3)^{(n-2)},(2n-2)^{(n-1)},(2n-1)^{(n-1)}\}$ in $\mathbf{P}_{n+2}$ than in $\mathbf{P}_{n}$ by \eqref{Podd}, \eqref{betan} and \eqref{gamman}. So there are $2n+4$ integers in $\mathbf{P}_{n+2}$, and for any integer $s \in[0,2n+4)$, the superscript is $(\lfloor \frac{s}{2} \rfloor)$, and the following statements hold.
\begin{itemize}
\item When $s=0$ or $5$, $s^{(\lfloor \frac{s}{2} \rfloor)}$ only appears in columns $[0,n)\setminus \{0,2\}$ of $\mathbf{P}_n $ by hypothesis. And $s^{(\lfloor \frac{s}{2} \rfloor)}$ appears in $\beta_n$ or $\gamma_n$ if $s$ is even or odd by \eqref{beta3}, \eqref{gamma3}, \eqref{betan} and \eqref{gamman}. So $s^{(\lfloor \frac{s}{2} \rfloor)}$ only appears in columns $[0,n+2)\setminus \{0,2\}$ of $\mathbf{P}_{n+2}$ by \eqref{Podd};
\item When $s=1$ or $2$, $s^{(\lfloor \frac{s}{2} \rfloor)}$ only appears in columns $[0,n)\setminus \{0,1\}$ of $\mathbf{P}_n $ by hypothesis. And $s^{(\lfloor \frac{s}{2} \rfloor)}$ appears in $\beta_n$ or $\gamma_n$ if $s$ is even or odd by \eqref{beta3}, \eqref{gamma3}, \eqref{betan} and \eqref{gamman}. So $s^{(\lfloor \frac{s}{2} \rfloor)}$ only appears in columns $[0,n+2)\setminus \{0,1\}$ of $\mathbf{P}_{n+2}$ by \eqref{Podd};
\item When $s=3$ or $4$, $s^{(\lfloor \frac{s}{2} \rfloor)}$ only appears in columns $[0,n)\setminus \{1,2\}$ of $\mathbf{P}_n $ by hypothesis. And $s^{(\lfloor \frac{s}{2} \rfloor)}$ appears in $\beta_n$ or $\gamma_n$ if $s$ is even or odd by \eqref{beta3}, \eqref{gamma3}, \eqref{betan} and \eqref{gamman}. So $s^{(\lfloor \frac{s}{2} \rfloor)}$ only appears in columns $[0,n+2)\setminus \{1,2\}$ of $\mathbf{P}_{n+2}$ by \eqref{Podd};
\item When $s\in[6,2n+4)$, let us consider the value of $s$.
     \begin{itemize}
     \item If $s\in[6,2n)$ and $\lfloor\frac{s}{2}\rfloor$ is even, $s^{(\lfloor\frac{s}{2}\rfloor)}$ only appears in columns $[0,n)\setminus \{\lfloor\frac{s}{2}\rfloor-1,\lfloor\frac{s}{2}\rfloor\}$ of $\mathbf{P}_n $ by hypothesis. And $s^{(\lfloor\frac{s}{2}\rfloor)}$ appears in $\beta_n$ or $\gamma_n$ if $s$ is even or odd by \eqref{betan} and \eqref{gamman}. So $s^{(\lfloor\frac{s}{2}\rfloor)}$ only appears in columns $[0,n+2)\setminus \{\lfloor\frac{s}{2}\rfloor-1,\lfloor\frac{s}{2}\rfloor\}$ of $\mathbf{P}_{n+2} $ by \eqref{Podd};
     \item If $s\in[6,2n)$ and $\lfloor\frac{s}{2}\rfloor$ is odd, $s^{(\lfloor\frac{s}{2}\rfloor)}$ only appears in columns $[0,n)\setminus \{\lfloor\frac{s}{2}\rfloor,\lfloor\frac{s}{2}\rfloor+1\}$ of $\mathbf{P}_n $ by hypothesis. And $s^{(\lfloor\frac{s}{2}\rfloor)}$ appears in $\beta_n$ or $\gamma_n$ if $s$ is even or odd by \eqref{betan} and \eqref{gamman}. So $s^{(\lfloor\frac{s}{2}\rfloor)}$ only appears in columns $[0,n+2)\setminus \{\lfloor\frac{s}{2}\rfloor,\lfloor\frac{s}{2}\rfloor+1\}$ of $\mathbf{P}_{n+2} $ by \eqref{Podd};
     \item If $s\in[2n,2n+4)$, $s^{(\lfloor\frac{s}{2}\rfloor)}$ only appears in $s^{(\lfloor\frac{s}{2}\rfloor)}\mathbf{I}_n$ of $\mathbf{P}_{n+2}$ by \eqref{Podd}. So $s^{(\lfloor\frac{s}{2}\rfloor)}$ only appears in columns $[0,n)=[0,n+2)\setminus \{n,n+1\}$.
     \end{itemize}
\end{itemize}
\end{proof}

\begin{proposition}
\label{proposition4}
For any odd integer $K\geq 3$, let $\mathbf{P}_{K}$ be recursively got from \eqref{Podd}. If $K=n+2$, $k\in\{n,n+1\}$ and $\mathbf{P}_{n+2}(i,k)=s^{(\lfloor \frac{s}{2} \rfloor)}$, then $s\in[0,2n)$ and the following statements hold.
\begin{itemize}
\item When $s\in\{2,5\}$, the entries in row $i$ and columns $[0,n)\setminus \{0\}$ of $\mathbf{P}_{n+2}$ are all $``*"$s;
\item When $s\in\{1,4\}$, the entries in row $i$ and columns $[0,n)\setminus \{1\}$ of $\mathbf{P}_{n+2}$ are all $``*"$s;
\item When $s\in\{0,3\}$, the entries in row $i$ and columns $[0,n)\setminus \{2\}$ of $\mathbf{P}_{n+2}$ are all $``*"$s;
\item When $s\in[6,2n)$ and $\lfloor\frac{s}{2}\rfloor$ is odd, the entries in row $i$ and columns $[0,n)\setminus \{\lfloor\frac{s}{2}\rfloor+1\}$ of $\mathbf{P}_{n+2}$ are all $``*"$s;
\item When $s\in[6,2n)$ and $\lfloor\frac{s}{2}\rfloor$ is even, the entries in row $i$ and columns $[0,n)\setminus \{\lfloor\frac{s}{2}\rfloor-1\}$ of $\mathbf{P}_{n+2}$ are all $``*"$s.
\end{itemize}
\end{proposition}
\begin{proof}
Since $\mathbf{P}_{n+2}(i,k)=s^{(\lfloor \frac{s}{2} \rfloor)}$ and $k\in\{n,n+1\}$, then from \eqref{Podd} we have $s\in[0,2n)$, and
\begin{itemize}
\item when $s\in\{2,5\}$, the $i$-th row and columns $[0,n)\cup \{k\}$ of $\mathbf{P}_{n+2}$ is as follows
\begin{equation*}
      \bordermatrix{%
                 & 0        & 1 &\cdots &  n-1 & k      \cr
              i  &s'^{(s')} & * &\cdots & *    &  s^{(s)} };
\end{equation*}
\item when $s\in\{1,4\}$, the $i$-th row and columns $[0,n)\cup \{k\}$ of $\mathbf{P}_{n+2}$ is as follows
\begin{equation*}
      \bordermatrix{%
                 & 0 & 1         & 2 & \cdots  & n-1  & k       \cr
              i  & * & s'^{(s')} & * &  \cdots & *    & s^{(s)} };
\end{equation*}
\item when $s\in\{0,3\}$, the $i$-th row and columns $[0,n)\cup \{k\}$ of $\mathbf{P}_{n+2}$ is as follows
\begin{equation*}
      \bordermatrix{%
                 & 0 & 1 & 2       & 3 &\cdots & n-1 & k       \cr
              i  & * & * &s'^{(s')}& * &\cdots & *   & s^{(s)} };
\end{equation*}
\item when $s\in[6,2n)$ and $\lfloor\frac{s}{2}\rfloor$ is odd, the $i$-th row and columns $[0,n)\cup \{k\}$ of $\mathbf{P}_{n+2}$ is as follows
\begin{equation*}
      \bordermatrix{%
                 & 0 & 1 &\cdots & \lfloor\frac{s}{2}\rfloor & \lfloor\frac{s}{2}\rfloor+1   & \lfloor\frac{s}{2}\rfloor+2    & \cdots & n-1 & k       \cr
              i  & * & * &\cdots & *   &  s'^{(s')} & *    & \cdots & *   & s^{(s)} };
\end{equation*}
\item when $s\in[6,2n)$ and $\lfloor\frac{s}{2}\rfloor$ is even, the $i$-th row and columns $[0,n)\cup \{k\}$ of $\mathbf{P}_{n+2}$ is as follows
\begin{equation*}
      \bordermatrix{%
                 & 0 & 1 &\cdots & \lfloor\frac{s}{2}\rfloor-2 & \lfloor\frac{s}{2}\rfloor-1   & \lfloor\frac{s}{2}\rfloor    & \cdots & n-1 & k       \cr
              i  & * & * &\cdots & *   &  s'^{(s')} & *    & \cdots & *   & s^{(s)} }
\end{equation*}
\end{itemize}
where $s'$ is some integer in $[2n,2n+4)$. So the proposition is true.
\end{proof}

Now let us propose the proof of Theorem \ref{th-K_ODD_K-2/K}.
\begin{proof}
First we use mathematic induction to prove that $\mathbf{P}_K$ recursively got from \eqref{Podd} is a $(K, 1, K(K-2), (K-2)^2, 2K)$ DPDA. When $K=3$, $\mathbf{P}_3$ defined by \eqref{P3} is a $(3,1,3,1,6)$ DPDA. When $K=n$, assume that $\mathbf{P}_{n}$ is a $(n, 1, n(n-2), (n-2)^2, 2n)$ DPDA. When $K=n+2$, we should consider the following cases.
\begin{itemize}
\item $\mathbf{P}_{n+2}$ is a $\left(n(n-2)+4n\right) \times (n+2)$ array by \eqref{Podd}, so it has $F=n(n+2)$ rows and $K=n+2$ columns.
\item By \eqref{Podd} there there are $(n-2)^2+4(n-1)=n^2$ $``*"$s in column $k$ of $\mathbf{P}_{n+2}$ for any $k\in[0,n)$, and there are $n(n-2)+2n=n^2$ $``*"$s in column $k$ of $\mathbf{P}_{n+2}$ for any $k\in \{n,n+1\}$. So there are $Z=n^2$ $``*"$s in each column of $\mathbf{P}_{n+2}$. C$1$ holds.
\item There are $2n+4$ integers in $\mathbf{P}_{n+2}$, and for any integer $s\in[0:2n+4)$, the superscript is $(\lfloor \frac{s}{2} \rfloor)$ by Proposition \ref{proposition3}. C$2$ holds.
\item When $\mathbf{P}_{n+2}(i,j)=s^{(\lfloor\frac{s}{2}\rfloor)}$, let us consider the value of $s$.
    \begin{itemize}
    \item If $s\in[0,2n)$ and $j\in[0,n)$, $\mathbf{P}_{n+2}(i,\lfloor\frac{s}{2}\rfloor)=\mathbf{P}_{n}(i,\lfloor\frac{s}{2}\rfloor)=*$ by \eqref{Podd}  and hypothesis;
    \item If $s\in[0,2n)$ and $j\in \{n,n+1\}$, $\mathbf{P}_{n+2}(i,\lfloor\frac{s}{2}\rfloor)=*$ by Proposition \ref{proposition4};
    \item If $s\in[2n,2n+4)$, it's easy to see that $\mathbf{P}_{n+2}(i,\lfloor\frac{s}{2}\rfloor)=*$ by \eqref{Podd}.
    \end{itemize}
So C$3$ holds.
\item  For any two distinct entries $\mathbf{P}_{n+2}(i_1,j_1)$ and $\mathbf{P}_{n+2}(i_2,j_2)$, assume that $\mathbf{P}_{n+2}(i_1,j_1)=\mathbf{P}_{n+2}(i_2,j_2)=s^{(\lfloor\frac{s}{2}\rfloor)}$.
    \begin{itemize}
    \item When $s\in [2n,2n+4)$, $s^{(\lfloor\frac{s}{2}\rfloor)}$ only appears in $s^{(\lfloor \frac{s}{2} \rfloor)}\mathbf{I}_n$ by \eqref{Podd}, so $\mathbf{P}_{n+2}(i_2,j_1)=\mathbf{P}_{n+2}(i_1,j_2)=*$;
    \item When $s\in [0,2n)$, $\lfloor\frac{s}{2}\rfloor \in[0,n)$, let us consider the values of $j_1$ and $j_2$.
         \begin{itemize}
         \item[$1)$] If $j_1,j_2\in[0,n)$, $\mathbf{P}_{n+2}(i_2,j_1)=\mathbf{P}_{n}(i_2,j_1)=*$ and $\mathbf{P}_{n+2}(i_1,j_2)=\mathbf{P}_{n}(i_1,j_2)=*$ by \eqref{Podd} and hypothesis;
         \item[$2)$] If $j_1,j_2\in \{n,n+1\}$, $\mathbf{P}_{n+2}(i_2,j_1)=\mathbf{P}_{n+2}(i_1,j_2)=*$ by \eqref{Podd};
         \item[$3)$] If one of $j_1,j_2$ is in $[0,n)$, and the other is in $\{n,n+1\}$, without loss of generality, assume that $j_1<j_2$. Then $j_1\in [0,n)$ and $j_2\in \{n,n+1\}$. So $\mathbf{P}_{n+2}(i_1,j_1)$ is in $\mathbf{P}_{n}$ and $\mathbf{P}_{n+2}(i_2,j_2)$ is in $\beta{_n^ \mathrm{ T }}$ or $\gamma{_n^ \mathrm{ T }}$ if $s$ is even or odd. Hence, we have $\mathbf{P}_{n+2}(i_1,j_2)=*$ by \eqref{Podd}.
             \begin{itemize}
             \item If $s=0$ or $5$, $s^{(\lfloor\frac{s}{2}\rfloor)}$ only appears in columns $[0,n)\setminus \{0,2\}$ of $\mathbf{P}_{n}$ by Proposition \ref{proposition3}. Then $j_1\neq 0,2$. So $\mathbf{P}_{n+2}(i_2,j_1)=*$ by Proposition \ref{proposition4};
             \item If $s=1$ or $2$, $s^{(\lfloor\frac{s}{2}\rfloor)}$ only appears in columns $[0,n)\setminus \{0,1\}$ of $\mathbf{P}_{n}$ by Proposition \ref{proposition3}. Then $j_1\neq 0,1$. So $\mathbf{P}_{n+2}(i_2,j_1)=*$ by Proposition \ref{proposition4};
             \item If $s=3$ or $4$, $s^{(\lfloor\frac{s}{2}\rfloor)}$ only appears in columns $[0,n)\setminus \{1,2\}$ of $\mathbf{P}_{n}$ by Proposition \ref{proposition3}. Then we have $j_1\neq 1,2$. So $\mathbf{P}_{n+2}(i_2,j_1)=*$ by Proposition \ref{proposition4};
             \item If $s\in[6,2n)$  and $\lfloor\frac{s}{2}\rfloor$ is even, $s^{(\lfloor\frac{s}{2}\rfloor)}$ only appears in columns $[0,n)\setminus \{\lfloor\frac{s}{2}\rfloor-1,\lfloor\frac{s}{2}\rfloor\}$ of $\mathbf{P}_{n}$ by Proposition \ref{proposition3}. Then $j_1\neq \lfloor\frac{s}{2}\rfloor-1$. So $\mathbf{P}_{n+2}(i_2,j_1)=*$ by Proposition \ref{proposition4};
             \item If $s\in[6,2n)$  and $\lfloor\frac{s}{2}\rfloor$ is odd, $s^{(\lfloor\frac{s}{2}\rfloor)}$ only appears in columns $[0,n)\setminus \{\lfloor\frac{s}{2}\rfloor,\lfloor\frac{s}{2}\rfloor+1\}$ of $\mathbf{P}_{n}$ by Proposition \ref{proposition3}. Then $j_1\neq \lfloor\frac{s}{2}\rfloor+1$. So $\mathbf{P}_{n+2}(i_2,j_1)=*$ by Proposition \ref{proposition4}.
             \end{itemize}

        \end{itemize}

    \end{itemize}
  So C$4$ holds.

\end{itemize}

Hence, $\mathbf{P}_{n+2}$ is a $(n+2,1, n(n+2), n^2, 2n+4)$ DPDA. So $\mathbf{P}_K$ recursively got from \eqref{Podd} is a $(K, 1, K(K-2), (K-2)^2, 2K)$ DPDA. Since $K\geq 3$ is odd, let $K=2q+1,q\geq 1$. Then $\mathbf{P}_K$ is a $(2q+1, 1, 4q^2-1, (2q-1)^2, 4q+2)$ DPDA.
\end{proof}


\begin{thebibliography}{1}
\bibitem{CVNI}
Cisco visual networking index: Global mobile data traffic forecast update, 2016-2021 White Paper.

\bibitem{KM}
K. C. Almeroth and M. H. Ammar, The use of multicast delivery to provide a scalable and interactive video-on-demand service, {\em IEEE Journal on Selected Areas in Communications}, vol. 14, no. 6, pp. 1110-1122, 1996.

\bibitem{AG}
M. M. Amiri and D. G\"{u}nd\"{u}z, Fundamental limits of caching: Improved delivery rate-cache capacity trade-off, {\em IEEE Transactions on Communications}, vol. 65, no. 2, pp. 806-815, 2017.

\bibitem{AQV}
A. Asadi, Q. Wang, and V. Mancuso, A survey on Device-to-Device communication in cellular networks, {\em IEEE Communications Surveys and Tutorials}, 2014, vol. 16, no. 4, pp. 1801-1819.

\bibitem{IRC}
I. D. Baev, R. Rajaraman, and C. Swamy, Approximation algorithms for data placement problems, {\em SIAM Journal on Computing}, vol. 38, no. 4, pp. 1411-1429, 2008.

\bibitem{ZYTT}
Z. Bar-Yossef, Y. Birk, T. Jayram, and T. Kol, Index coding with side information, {\em IEEE Transactions on Information Theory}, vol. 57, no. 3, pp. 1479-1494, 2011.

\bibitem{SVA}
S. C. Borst, V. Gupta, and A. Walid, Distributed caching algorithms for content distribution networks, in Proceedings of
{\em 29th International Conference on Computer Communications}, 15-19 March 2010, San Diego, CA, USA, pp. 1478-1486.

\bibitem{CYTJ}
M. Cheng, Q. Yan, X. Tang, and J. Jiang ,Coded caching schemes with low rate and subpacketizations, arXiv:1703.01548v2 [cs.IT], Apr 2017.

\bibitem{ADP}
A. Dan, D. Sitaram, and P. Shahabuddin, Dynamic batching policies for an on-demand video server, {\em Multimedia Systems}, vol. 4, no. 3, pp. 112-121, 1996.

\bibitem{LD}
L. W. Dowdy and D. V. Foster, Comparative models of the file assignment problem, {\em ACM Computing Sruverys}, vol. 14, no. 2, pp. 287-313, Jun. 1982

\bibitem{GR}
H. Ghasemi and A. Ramamoorthy, Improved lower bounds for coded caching, in Proc. {\em IEEE International Symposium on Information Theory}, Hong Kong, Jun. 2015, pp. 1696-1700.

\bibitem{NAAG}
N. Golrezaei, A. F. Molisch, A. G. Dimakis, and G. Caire, Femtocaching and device-to-device collaboration: A new architecture for wireless video distribution, {\em IEEE Communications Magazine}, vol. 51, no. 4, pp. 142-149, 2013.

\bibitem{MGA}
M. Ji, G. Caire, and A.F. Molisch, Fundamental limits of caching in wireless D2D networks, {\em IEEE Transactions on Information Theory}, 2016, vol. 62, no. 2, pp. 849-869.

\bibitem{MCR}
M. R. Korupolu, C. G. Plaxton, and R. Rajaraman, Placement algorithms for hierarchical cooperative caching, in Proceedings of the {\em 10th Annual ACM-SIAM Symposium on Discrete Algorithms}, 17-19 January 1999, Baltimore, Maryland., 1999, pp. 586-595.

\bibitem{YY}
Y.-D. Lin and Y.-C. Hsu, Multihop cellular: A new architecture for wireless communications, in Proceedings of
{\em 19th International Conference on Computer Communications}, 26-30 March 2000, Tel Aviv, Israel, pp. 1273-1282.

\bibitem{MU}
M. A. Maddah-Ali and U. Niesen, Fundamental limits of caching, {\em IEEE Transactions on Information Theory}, vol. 60, no. 5, pp. 2856-2867, 2014.


\bibitem{AKS}
A. Meyerson, K. Munagala, and S. A. Plotkin, Web caching using access statistics, in Proceedings of the {\em 20th Annual Symposium on Discrete Algorithms}, January 7-9, 2001, Washington, DC, USA, pp. 354-363.

\bibitem{STC}
A. Sengupta, R. Tandon, and T. C. Clancy, Improved approximation of storage-rate tradeoff for caching via new outer
bounds, in Proc. {\em IEEE International Symposium on Information Theory}, Hong Kong, Jun. 2015, pp. 1691-1695.

\bibitem{T}
 C. Tian and J. Chen, Caching and delivery via interference elimination, in Proc. {\em IEEE International Symposium on Information Theory}, Barcelona, July 2016, pp. 830-834.

\bibitem{WLG}
C. Y. Wang, S. H. Lim, M. Gastpar, A new converse bound for coded caching, in Proc. {\em IEEE Information Theory Workshop}, Robinson College, Oct. 2016.

\bibitem{YTC}
Q. Yan,  X. Tang, and Q. Chen, On the gap between decentralized and centralized coded caching schemes, arXiv: 1605.04626 [cs.IT], May  2016.
\bibitem{QMXQ}
Q. Yan, M. Cheng, X. Tang and Q. Chen, On the placement delivery array design for centralized coded caching scheme, {\em IEEE Transactions on Information Theory}, vol. 63, no. 9, pp. 5821-5833, 2017.

\end{thebibliography}
\end{document}